\documentclass[journal,twoside,web]{ieeecolor}
\usepackage{generic}
\usepackage{cite}
\usepackage{amsmath,amssymb,amsfonts}
\usepackage{algorithmic}
\usepackage{graphicx}
\usepackage{algorithm,algorithmic}
\usepackage{hyperref}

\hypersetup{hidelinks=true}
\usepackage{textcomp}

\usepackage{xcolor}
\usepackage{mathtools}
\usepackage{bm}

\usepackage{tikz}
\usetikzlibrary{fit}
\usetikzlibrary{arrows} 
\usetikzlibrary{decorations.pathmorphing}
\usetikzlibrary{matrix} 
\usetikzlibrary{calc}
\usetikzlibrary{backgrounds}

\newcommand{\KCBF}{\mathcal{U}_{h}}
\newcommand{\hpb}{h_{\mathrm{PB}}}
\newcommand{\ahpb}{\hat{h}_{\mathrm{PB}}}

\newcommand{\Spb}{\mathcal{S}_{\mathrm{PB}}}
\newcommand{\Dpb}{\mathcal{D}_{\mathrm{PB}}}

\DeclareMathOperator{\interior}{int}

\newtheorem{theorem}{Theorem}
\newtheorem{lemma}{Lemma}
\newtheorem{proposition}{Proposition}
\newtheorem{definition}{Definition}
\newtheorem{assumption}{Assumption}
\newtheorem{remark}{Remark}
\newtheorem{corollary}{Corollary}

\def\BibTeX{{\rm B\kern-.05em{\sc i\kern-.025em b}\kern-.08em
    T\kern-.1667em\lower.7ex\hbox{E}\kern-.125emX}}
\begin{document}
\title{Approximate predictive control barrier function for discrete-time systems}
\author{Alexandre Didier, \IEEEmembership{Graduate Student Member, IEEE}, and Melanie N. Zeilinger, \IEEEmembership{Member, IEEE}
\thanks{The authors are with the Institute for Dynamic Systems and Control, ETH Zurich, Zurich 8092, Switzerland. (e-mail: \{adidier,mzeilinger\}@ethz.ch).}
}

\maketitle
\begin{abstract}
We propose integrating an approximation of a predictive control barrier function (PCBF) in a safety filter framework, resulting in a prediction horizon independent formulation. The PCBF is defined through the value function of an optimal control problem and ensures invariance as well as stability of a safe set within a larger domain of attraction. We provide a theoretical analysis of the proposed algorithm, establishing input-to-state stability of the safe set with respect to approximation errors as well as exogenous disturbances. Furthermore, we propose a continuous extension of the PCBF within the safe set, reducing the impact of learning errors on filter interventions. We demonstrate the stability properties and computational advantages of the proposed algorithm on a linear system example and its application as a safety filter for miniature race cars in simulation.
\end{abstract}

\section{INTRODUCTION}
\label{sec:introduction}
Safety filters provide a modular mechanism which enables certifying and, if necessary, modifying any control input such that safety in terms of constraint satisfaction can be guaranteed. The filters are designed to be agnostic with respect to the main control performance objective and therefore aim to minimise the distance between the proposed input and the filtered input which is applied to the system. Thereby, safety filters are suitable tools to enhance the safety properties of various systems. The filters can be used to augment human agents operating systems, such as, e.g., automotive vehicles as demonstrated on miniature race cars in \cite{tearle2021}, and have also demonstrated their effectiveness in complex robotics applications, such as, e.g., walking on stepping stones with a legged robot in \cite{grandia2021}.

The core principle of safety filters consists in obtaining a control invariant set within the state constraints, such that inputs applied to the system guarantee that the system state remains within this invariant set. In order to obtain such a control invariant set, three main branches have been pursued, as discussed extensively in \cite{wabersich2023}. 
Hamilton-Jacobi reachability analysis, as discussed in \cite{chen2018}, allows computing the maximal control invariant subset to constraints, however the method suffers from poor computational scalability with respect to increasing state dimension.
Control barrier functions (CBFs) aim to alleviate the computational restrictions, by explicitly computing, potentially smaller, control invariant sets. An overview is provided in \cite{ames2019}. CBFs are obtained either through either expert knowledge or more scalable computational methods compared to the Hamilton-Jacobi reachability, such as semi-definite programming or sum-of-squares optimisation, such as in \cite{clark2021}. 
Finally, predictive safety filters have been proposed in \cite{wabersich2018}. They employ model predictive control (MPC) techniques in order to implicitly define a control invariant set through an online predictive optimisation problem. A detailed overview of such MPC techniques can be found in \cite{rawlings2017}. These predictive methods make use of discrete-time dynamics formulations and allow to easily provide recursive feasibility guarantees. They enable extensions to, e.g., robust treatments of disturbances, stochastic disturbances in \cite{wabersich21} or using online data for uncertain systems to improve safety filter performance in \cite{didier2021}, due to the well developed MPC theory. 

In practical systems, considering hard constraints in predictive control methods often leads to infeasibility due to unmodelled dynamics or sensor noise. Therefore, the state constraints on the system are often softened, as proposed in~\cite{kerrigan00} and used in~\cite{tearle2021}, by introducing slack variables. However, this generally leads to a loss of theoretical properties as soon as the system violates the constraints. In \cite{wabersich2022}, a constraint softening mechanism is proposed, which is able to provide theoretical guarantees when constraints are violated. Connections of this mechanism to CBFs are investigated and it is shown that by minimising a function which depends on the slack variables, a control barrier function is implicitly defined. This result implies that the implicitly defined safe set of a predictive safety filter is stable within a larger domain of attraction, such that the system is able to recover from unexpected disturbances or model mismatch. The main drawback of the method in \cite{wabersich2022} is the need to solve two predictive optimisation problems in order to guarantee stability, which results in a computationally demanding algorithm. 

\textit{Contributions.} In this work, we leverage an explicit approximation of the PCBF in a safety filter formulation. 
 This formulation requires predicting only a single time step ahead, allowing for efficient online computations. 
 The contributions compared to initial results in~\cite{didier2023} are as follows:
 \begin{enumerate}
     \item We extend initial input-to-state stability (ISS) analysis in~\cite{didier2023}, with respect to approximation errors, to the case of additional exogenous disturbances. This is achieved through a relation between control barrier function (CBF) definitions based on continuity and class $\mathcal{K}$ comparison functions, which enables a general robust set stability analysis. The provided analysis directly implies ISS of the PCBF algorithm in~\cite{wabersich2022}.
     \item We propose a novel continuous extension of the PCBF, leading to a significant reduction in safety filter interventions when the system is safe compared to~\cite{didier2023}.
     \item We introduce a hyperparameter that reduces the required CBF decrease in closed-loop compared to~\cite{didier2023} and~\cite{wabersich2022}, while preserving stability guarantees. 
     \item We show how the explicit approximation can be used in order to certify inputs using a single function evaluation, resulting in large computational speedups compared to~\cite{didier2023, wabersich2022} when safe inputs are proposed.
 \end{enumerate}
 
\textit{Related Work.}
Safety filters have been extensively researched and a review of state-of-the-art methods can be found in \cite{wabersich2023}. The first main branch of research consists of control barrier function based safety filters, originally proposed in \cite{wieland2007} for continuous time systems. The work in \cite{agrawal2017} proposed discrete-time control barrier functions. However, the proposed formulation uses a barrier function which approaches infinity as the state approaches the boundary of the safe set, which is not suitable for investigating the stability of the safe set within a larger domain. In \cite{chen2021}, a backup CBF formulation is proposed, which allows using a backup controller to enhance an invariant safe set, by computing gradients of a CBF on the closed-loop trajectory of the backup controller. Thereby, any state which reaches the safe set under the backup controller is still considered safe, and the safe set is effectively augmented. The predictive safety filters in \cite{wabersich2021} and \cite{wabersich2022} show relations to this method, where the optimiser is free to choose all control inputs in the prediction horizon, rather than selecting them with a prior fixed policy. In \cite{didier2024predictive}, the predictive safety filter is augmented with a stability guarantees of the origin using a Lyapunov decrease constraint. Compared to these works, we combine the advantages of the existence of inputs which stabilise a safe set and render it invariant, while using a single CBF decrease constraint rather than a full prediction horizon, reducing the computational effort. Due to the discrete-time formulation of the predictive controllers, the resulting CBF optimisation problem is a nonlinear program rather than a quadratic program in \cite{chen2021}, however the prediction horizon is not actually included in the online optimisation, effectively reducing the number of constraints.
The proposed formulation in \cite{breeden2022} uses a look-ahead to determine time instances, where given a path of the system, a CBF constraint is violated. Thereby, the control actions are modified such that safety can be guaranteed. However, this method does not consider the case when the input is constrained.
Finally, \cite{wiltz2023} proposes a reachability-based safety filter formulation in continuous time, which makes use of a finite lookahead time and establishes connections to model predictive safety filter methods. However, the method uses numerical approximations of the gradients of the CBF through sampling to impose the CBF decrease and it is unclear how approximation errors can impact the closed-loop performance. In comparison, we provide a rigorous analysis of how approximation errors of the learned CBF impact the closed-loop behaviour of the system.

Another main branch of research for safety filters is Hamilton-Jacobi reachability analysis, see, e.g., \cite{chen2018}. In \cite{choi2021}, connections between viscosity solution of the Hamilton-Jacobi-Isaacs variational inequality, retrieving the maximum robust control invariant subset to given constraints and CBFs are established and it is shown how this solution can be used in a quadratic program. This result is extended in \cite{choi2023} using a forward reachability approach, which shows that any CBF is also a forward reachability solution within the maximum robust control invariant set. Compared to these results, we consider stability of the control invariant set implicitly defined through the predictive safety filter optimisation problem within a larger domain in a soft-constrained manner. Furthermore, if the prediction horizon of the predictive safety filter goes to infinity, the maximum control invariant set can be recovered, see, e.g., \cite[Definition 2.9]{kerrigan2000}.

Numerous works consider approximation methods to reduce the computational drawbacks of online predictive control methods. Works such as, e.g., \cite{hertneck2018, paulson2020, nubert2020} are able to guarantee closed-loop properties such as stability and robust constraint satisfaction with respect to input uncertainty. However, the safety filter policy is in general non-continuous and multi-valued, such that continuous approximations can lead to arbitrary approximation errors. Other works consider learning the optimal value function of an MPC, such as, e.g., \cite{tamar2017, karnchanachari2020} to increase performance rather than to address general safety considerations. An overview for learning safety certificates is provided in \cite{dawson2023} with a main focus on certificate synthesis for continuous-time systems. In \cite{robey2020}, a control barrier function is learned from expert demonstrations with theoretical guarantees in case of Lipschitz continuity. These results could in principle be applicable for the proposed approximation scheme, however Lipschitz continuity of the predictive control barrier function is in general not guaranteed. Additionally, the required optimization problem is computationally complex, especially for neural network approximations. A discussion on loss function design for neural network approximations is provided in \cite{so2024}. In comparison, we provide guarantees on closed-loop system behaviour even if approximation errors are still present in the learned CBF and additionally consider the effect of exogenous disturbances.
In \cite{he2023}, a state-action control barrier function is approximated which makes use of contractivity of a safe set to ensure invariance of the resulting safe set for small enough approximation errors with a convex online optimisation problem. This results in the need to sample over the state-action space compared to sampling over the state space as proposed in this paper. Additionally, stability of the safe set itself is not considered in \cite{he2023} and recursive feasibility only holds if the approximation error is small enough, whereas the proposed formulation ensures feasibility everywhere.

\textit{Structure.} We discuss preliminaries in Section~\ref{sec:CBF} and summarize the PCBF algorithm in~\cite{wabersich2022}, which is the basis for the approximation used in this paper, in Section~\ref{sec:PCBF}. In Section~\ref{sec:APCBF}, we detail the proposed algorithm and the continuous extension of the PCBF, and in Section~\ref{sec:Theory}, we provide a relationship between two types of CBF definitions as well as the theoretical closed-loop analysis of the proposed algorithm. In Section~\ref{sec:NE}, we validate the approach in two numerical examples.

\textit{Notation.} Given a vector $a\in\mathbb{R}^n$ and a scalar $b\in\mathbb{R}\cup\{{\infty}\}$, an inequality $a < b$ denotes an element-wise inequality of each element of $a$, i.e. $[a]_i<b$ for all $i=1,\dots,n$. The distance of a vector $x$ to a set $\mathcal{A}$ is given by $\Vert x\Vert_\mathcal{A}:=\inf_{y\in\mathcal{A}} \Vert x-y\Vert$. A vector of appropriate dimensions with all entries equal to 1 is denoted as $\mathbf{1}$. 

\section{PRELIMINARIES} \label{sec:CBF}
Consider the discrete-time, nonlinear dynamical system
\begin{equation}
\label{eq:sys}
x(k+1) = f(x(k),u(k), w(k)) \;\; \forall k \in \mathbb{N},
\end{equation}
with state $x\in\mathbb{R}^n$, input $u\in\mathbb{R}^m$, exogenous disturbance $w(k)\in\mathbb{R}^p$ and continuous dynamics $f:\mathbb{R}^n \times \mathbb{R}^m\times \mathbb{R}^p \rightarrow \mathbb{R}^n$. The system is subject to state and input constraints 
\begin{equation}\label{eq:constraints}
x(k) \in \mathcal{X}:=\{x\in \mathbb{R}^n \mid c_x(x) \leq 0\}, \;\;
u(k) \in \mathcal{U},
\end{equation}
with $c_x:\mathbb{R}^n\rightarrow \mathbb{R}^{n_x}$.
The input constraints are considered to arise from, e.g., physical actuator limitations, while the state constraints arise from system specifications, such as, e.g., a limit on the velocity of a car on a highway lane. In practical predictive optimisation based methods, such constraints are typically formulated as soft constraints, see, e.g., \cite{kerrigan00}, in order to preserve feasibility under unknown disturbances of the optimisation problems which are solved online. Accordingly, we define the soft-constrained set
\begin{equation*}
    \mathcal{X}(\xi) := \{x\in \mathbb{R}^n \mid c_x(x)\leq \xi \},
\end{equation*}
where $\xi\in\mathbb{R}^{n_x}$ is a slack variable.
The constraint sets satisfy the following assumption. 

\begin{assumption}[Constraint Sets]\label{ass:cont}
    The sets $\mathcal{U}$ and $\mathcal{X}(\xi)$ are compact for all $0\leq \xi < \infty$ and $c_x$ is continuous.
\end{assumption}
\par
The goal is to design a predictive safety filter for system \eqref{eq:sys}. If an initial state violates the constraints $\mathcal{X}$, the system states should converge back into $\mathcal{X}$ and an initial state within a safe set should remain therein.
In order to ensure such closed-loop stability and constraint satisfaction, we make use of the concept of control barrier functions. These functions typically define an invariant set, see, e.g. \cite{blanchini1999}, as a zero sublevel set and provide stability guarantees within some domain with respect to this invariant set. Therefore, they are often considered to be an extension of a classical Lyapunov function, where the origin is stabilised. CBFs have been extensively studied in continuous-time dynamical systems, see, e.g., \cite{ames2019} for an overview, but have also been proposed in discrete-time in \cite{agrawal2017}. In this work, we consider two definitions for discrete-time CBFs introduced in \cite{wabersich2022} and \cite{rawlings2017}. These definitions differ in the properties the CBF needs to satisfy and we establish their relationship in Section~\ref{sec:Theory} in order to analyse robust stability of the proposed algorithm. 
\begin{definition}[Control Barrier Function - Continuous~\cite{wabersich2022}]
\label{def:CBFcont}
A function $h: \mathcal{D} \to \mathbb{R}$ is a discrete-time CBF with safe set $\mathcal{S} = \{x \in \mathbb{R}^n \mid h(x) \leq 0 \} \subset \mathcal{D}\subset \mathbb{R}^n$
 if the following conditions hold:\begin{enumerate}
    \item[1)] $\mathcal{S}, \mathcal{D}  $ are compact and non-empty sets,
    \item[2)] $h(x)$ is continuous on $\mathcal{D}$,
    \item[3)] $\exists \Delta h : \mathcal{D} \to \mathbb{R}$, with $\Delta h$ continuous and $\Delta h(x)>0 $ for all $x \in \mathcal{D} \setminus \mathcal{S} $ such that
\end{enumerate}
\vspace{-0.3cm}
\begin{subequations}
\label{eq:rate_cond_cbf}
\begin{align}
\hspace{-0.2cm}
    &\hspace{0.7cm}\forall x \in \mathcal{D} \setminus \mathcal{S}: \nonumber\\ & \inf_{u\in \mathcal{U}} \{h(f(x,u,0))\mid f(x,u,0)\in\mathcal{D}\}\!-\! h(x) \!\leq\! - \Delta h(x) \label{eq:outside_cbf_condition},\\
    &\hspace{0.72cm}\forall x \in \mathcal{S} \: : \inf_{u\in \mathcal{U}}h(f(x,u,0)) \leq 0. \label{eq:inside_cbf_condition}
\end{align}
\end{subequations}
\end{definition}
This first definition requires that the CBF $h$ is continuous and defines a set $\mathcal{S}$ as a zero sublevel set which is control invariant through the condition \eqref{eq:inside_cbf_condition}. Additionally, for any state outside the safe set, i.e., any $x\in\mathcal{D}\setminus\mathcal{S}$, there exists a continuous, positive decrease $\Delta h(x)$ of the CBF value. 
We note that compared to other CBF definitions such as in~\cite{agrawal2017}, we consider the 0-sublevel set as the safe set rather than the superlevel set.
If inputs that satisfy~\eqref{eq:rate_cond_cbf} are selected, it was shown in~\cite{wabersich2022} that the set $\mathcal{S}$ is asymptotically stable in $\mathcal{D}$ for the nominal closed-loop system by leveraging continuity-based arguments. 
Inputs which satisfy the CBF conditions are considered safe to be applied to the given dynamical system.
\begin{definition}[Safe inputs]\label{def:safeinp}
Given a CBF $h$ according to Definition~\ref{def:CBFcont}, an input $u$ applied to system \eqref{eq:sys} at state $x$ is said to be \textit{safe} for the nominal system, if $u \in \KCBF(x) \subseteq \mathcal{U}$ with
\begin{align} \label{eq:cbf_ctrl_law_cond}
    \KCBF(x) \coloneqq \left\{ \begin{array}{ll} 
    U_1(x), & \text{if } x\in \mathcal{D} \setminus \mathcal{S},\\
    U_2(x), & \text{if } x\in \mathcal{S},
    \end{array}\right.
\end{align}
where $U_1(x) \coloneqq \{ u \in \mathcal{U} \mid h(f(x,u,0)) - h(x) \leq -\Delta h(x)\}$ and $U_2(x) \coloneqq \{ u \in \mathcal{U} \mid h(f(x,u,0)) \leq 0 \}$.
\end{definition}
In order to formalise the closed-loop property of applying safe inputs to a system, we make use of the notion of asymptotic stability of invariant sets.

\begin{definition}[Asymptotic stability \cite{wabersich2022}]\label{def:asymstab}
    Let $\mathcal{S}$ and $\mathcal{D}$ be non-empty, compact and positively invariant sets for the system $x(k+1)=f(x(k),\kappa(x(k)),0)$ with $\kappa:\mathbb{R}^n\rightarrow\mathbb{R}^m$ such that $\mathcal{S}\subset\mathcal{D}$. The set $\mathcal{S}$ is called an asymptotically stable set in $\mathcal{D}$, if for all $x(0)\in\mathcal{D}$, the following conditions hold:
    \begin{subequations}
        \begin{align*}
        \forall \epsilon >0 \;\exists \delta >0: & \; \Vert x(0)\Vert_\mathcal{S} <\delta \Rightarrow \forall k>0: \Vert x(k)\Vert_\mathcal{S} < \epsilon, \\
        \exists \bar{\delta}>0: & \; \Vert x(0)\Vert_\mathcal{S} < \bar{\delta} \Rightarrow \lim_{k\rightarrow \infty} x(k)\in\mathcal{S}
        \end{align*}
    \end{subequations}
\end{definition}

Applying safe inputs to system \eqref{eq:sys} ensures that the set $\mathcal{S}$ is asymptotically stable in the domain $\mathcal{D}$ for the nominal system as given by the following result.

\begin{theorem}[\cite{wabersich2022}]
    Let $\mathcal{D}\subset\mathbb{R}^n$ be a non-empty and compact set. Consider a control barrier function $h:\mathcal{D}\rightarrow\mathbb{R}$ according to Definition~\ref{def:CBFcont} with $\mathcal{S}=\{x\in\mathbb{R}^n \mid h(x)\leq 0\}\subset \mathcal{D}$. If $\mathcal{D}$ is a forward invariant set for system $x(k+1)=f(x(k),u(k),0)$ under $u(k)=\kappa(x(k))$ for all $\kappa:\mathcal{D}\rightarrow\mathcal{U}$ with $\kappa(x)\in \KCBF(x)$, then it holds that
    \begin{enumerate}
        \item $\mathcal{S}$ is a forward invariant set
        \item $\mathcal{S}$ is asymptotically stable in $\mathcal{D}$
    \end{enumerate}
\end{theorem}

The second definition of a control barrier function which we consider is given in \cite[Appendix B.39]{rawlings2017} and makes use of class $\mathcal{K}$ comparison functions\footnote{A function $\alpha:[0,c]\rightarrow\mathcal{R}$ is class $\mathcal{K}$ if it is continuous, strictly increasing, $\alpha(0)=0$.}. 

\begin{definition}[Control barrier function - Class $\mathcal{K}$~\cite{rawlings2017}]
\label{def:CBFclK}
A function $h:\mathcal{D}\rightarrow \mathbb{R}$ is a discrete-time CBF with safe set $\mathcal{S} = \{x \in \mathbb{R}^n \mid h(x) \leq 0 \} \subset \mathcal{D}$
 if the following hold:\begin{enumerate}
    \item[1)] $\mathcal{S}, \mathcal{D}  $ are compact and non-empty sets,
    \item[2)] $\exists \; \alpha_1, \alpha_2\in\mathcal{K}$, such that $\alpha_1(\Vert x\Vert_\mathcal{S})\leq h(x)\leq \alpha_2(\Vert x\Vert_\mathcal{S})$,
    \item[3)] $\exists \; \alpha_3\in\mathcal{K}$ such that 
\end{enumerate}
    \begin{equation*}
        \inf_{u\in\mathcal{U}} \{h(f(x,u,0)) \mid f(x,u,0)\in\mathcal{D}\} - h(x) \leq - \alpha_3(\Vert x\Vert_\mathcal{S}).
    \end{equation*}
\end{definition}

\begin{remark}
    In \cite[Appendix B.39]{rawlings2017}, a function $h$ which fulfils Definition \ref{def:CBFclK} is denoted a control Lyapunov function, however in the scope of this paper we distinguish a control Lyapunov function as function which fulfils Definition~\ref{def:CBFclK} with $\mathcal{S}=\{0\}$ from a CBF with a general compact set $\mathcal{S}$. 
\end{remark}

A definition based on such comparison functions is advantageous due to the ease of extending stability guarantees if there are additional exogenous disturbances acting on the system in \eqref{eq:sys}. In Section~\ref{sec:Theory}, we will show that a CBF fulfilling Definition~\ref{def:CBFcont} implies existence of a CBF according to Definition~\ref{def:CBFclK}. A similar relation of a continuous Lyapunov decrease implying the existence of a class $\mathcal{K}$ decrease is shown in \cite{jiang2002}, but uses class $\mathcal{K}_\infty$ bounds on the original Lyapunov function rather than the continuity of $h$ in Definition~\ref{def:CBFcont}. The established relation between the two CBF definitions allows showing the input-to-state stability guarantees of the approximate predictive CBF approach proposed in this paper.

\section{PRELIMINARIES ON PREDICTIVE CONTROL BARRIER FUNCTIONS} \label{sec:PCBF}
In this section, we review the concept of predictive safety filters, as proposed originally in~\cite{wabersich2018} and adapted in~\cite{wabersich2022} for system recovery, and how they can be related to the concept of CBFs. At every time step, a controller proposes a potentially unsafe control action, denoted as $u_\mathrm{p}(k)$. Such a controller could be optimising a primary control objective, such as a robot operating in a cluttered environment or even a human operating a car. The aim of the safety filter is to certify whether the proposed input is safe to apply to the system in terms of the state and input constraints~\eqref{eq:constraints}, and modify it, if necessary. In order to guarantee constraint satisfaction for all times, the predictive safety filter makes use of model predictive control techniques. 
The predictive safety filter problem proposed in~\cite{wabersich2022} which is solved online at every time step is given as
\begin{subequations}
\label{eq:pcbf_sf}
\begin{align}
\min_{x_{i|k}, u_{i|k}} \; \; & \lvert \lvert{u_{0|k} - u_\mathrm{p}(k)}\rvert \rvert \\
\text{s.t.}\; \; &\forall\, i = 0,\dots,N-1, \nonumber \\
& x_{0|k} = x(k), \\
& x_{i+1|k} = f(x_{i|k}, u_{i|k}, 0), \label{eq:pcbf_dyn}\\
& u_{i|k} \in \mathcal{U}, \label{eq:pcbf_sf_inputconstraints} \\
& x_{i|k} \in \overline{\mathcal{X}}_i( \xi^*_{i|k}), \label{eq:pcbf_sf_stateconstraints} \\
&h_f(x_{N|k}) \leq  \xi^*_{N|k}, \label{eq:pcbf_sf_terminalconstraints}
\end{align}
\end{subequations}
Given the current state $x(k)$, a nominal trajectory of states $x_{i|k}$ and inputs $u_{i|k}$ is planned according to the system dynamics in~\eqref{eq:pcbf_dyn} such that the state and input constraints~\eqref{eq:pcbf_sf_inputconstraints} and~\eqref{eq:pcbf_sf_stateconstraints} are satisfied for all predicted time steps. We note that as the state can lie outside the set $\mathcal{X}$, the state constraints are formulated in terms of soft constraints with respect to slack variables $\xi^*_{i|k}$ which are computed according to a second optimisation problem discussed later in this section. The state constraints are subject to an additional tightening, i.e., $\overline{\mathcal{X}}_i(\xi_{i|k}):=\{x\in\mathbb{R}^n \mid c_x(x) \leq -\Delta_i + \xi_{i|k}\}$ with $\Delta_i>\Delta_{i-1}>\Delta_0=0$, in order to establish stability guarantees of the closed-loop. A terminal constraint~\eqref{eq:pcbf_sf_terminalconstraints} is imposed on the final predicted state, where $h_f:\mathcal{D}_f \rightarrow \mathbb{R}$ is a CBF according to Definition~\ref{def:CBFcont} and is assumed to be known.
\begin{assumption}[Terminal control barrier function]~ \label{ass:terminalCBF}
 A terminal CBF $h_f:\mathcal{D}_f \rightarrow \mathbb{R}$ with domain $\mathcal{D}_f:=\{x\in\mathbb{R}^n \mid h_f(x) \leq \gamma_f \}$ for $\gamma_f>0$ and corresponding safe set $\mathcal{S}_f:=\{x\in\mathbb{R}^n \mid h_f(x) \leq 0 \}\subseteq \overline{\mathcal{X}}_{N-1}(0)$ exists according to Definition~\ref{def:CBFcont} for the dynamics~\eqref{eq:sys} and is known. 
\end{assumption} 

A number of design procedures exist to design such a CBF, including Lyapunov-based methods such as semi-definite programming for quadratic functions, e.g., in \cite{vanantwerp2000}, which can also be applied to nonlinear systems as discussed, e.g., in \cite[Section IV]{wabersich2022}, and sum-of-squares optimisation for polynomial systems and Lyapunov functions in, e.g., \cite{jarvis2003}. While obtaining a control barrier function for nonlinear systems is a challenging problem in general, we note that a continuous, local Lyapunov function around an equilibrium point is sufficient to obtain $h_f$ satisfying Definition~\ref{def:CBFcont} as discussed in \cite[Section IV]{wabersich2022} and conservativeness introduced by small terminal sets $\mathcal{S}_f$ can be counteracted by increasing the prediction horizon $N$ in the optimisation problem \eqref{eq:pcbf_sf}.

In order to be minimally invasive, given the proposed performance input $u_\mathrm{p}(k)$ at any given time step, the safety filter problem minimises the distance of the first predicted input $u_{0|k}$ and $u_\mathrm{p}(k)$. Trivially, if the proposed performance input $u_\mathrm{p}(k)$ allows constructing a safe backup trajectory, it is also the optimal solution to the safety filter problem. At every time step, after obtaining a solution to \eqref{eq:pcbf_sf}, the input $u(k)=u^*_{0|k}$ is applied to the system.

Due to the recursive feasibility property shown, e.g., in \cite{wabersich2022}, it holds that the predictive safety filter \eqref{eq:pcbf_sf} implicitly defines an invariant subset of the state constraints, i.e. when $\xi^*_{i|k}=0$. Any state which can reach terminal invariant set within the prediction horizon such that state and input constraints are satisfied is safe.
Thereby, the predictive safety filter can be seen as an enhancement mechanism for invariant sets using model predictive control techniques. Trivially, any state within the invariant set $\mathcal{S}_f\coloneqq\{x\in\mathbb{R}^n\mid h_f(x)\leq 0\}$ is a feasible point in the optimisation problem $\eqref{eq:pcbf_sf}$, which implies that the set of feasible points with $\xi^*_{i|k}=0$ is a superset of $\mathcal{S}_f$, for any given prediction horizon $N \geq 1$.
Such a recovery mechanism is achieved in~\cite{wabersich2022} by computing the slack variables $\xi^*_{i|k}$ in~\eqref{eq:pcbf_sf} through the optimisation problem

Building upon the introduced predictive safety filter formulation, it is shown in \cite{wabersich2022}, that a similar enhancement mechanism can be formulated for CBFs to increase the domain of attraction of the closed-loop system. As in practical applications, the initial state $x(0)$ may not be feasible without soft constraints or unmodelled external disturbances could result in infeasible states for the optimisation problem, a recovery mechanism is required, which ensures that the system converges back to an invariant subset of the constraints $\mathcal{X}$ in such events. Through their property of stability of invariant sets in a given domain, CBFs offer a theoretical framework which is suitable for such a recovery mechanism. Such an enhancement property for the domain of attraction of the closed-loop system is achieved by computing the slack variables $\xi^*_{i|k}$ in \eqref{eq:pcbf_sf} according to the following optimisation problem, which minimises the state constraint violations and a weighted penalty on the terminal slack $\xi_{N|k}$ with $\alpha_f>0$, i.e.
\begin{subequations}
\label{eq:pcbf}
\begin{align}
 \hspace{-0.1cm}\hpb(x(k)) := \min_{x_{i|k},u_{i|k}, \xi_{i|k}} \; \; & \alpha_f \xi_{N|k} + \sum_{i=0}^{N-1}\lvert \lvert{\xi_{i|k}}\rvert \rvert \label{eq:cbf_pcbf} \\
\text{s.t.}\; \; &\forall\, i = 0,\dots,N-1, \nonumber \\
& x_{0|k} = x(k), \label{eq:slack_init_pred} \\
& x_{i+1|k} = f(x_{i|k}, u_{i|k}, 0), \\
& u_{i|k} \in \mathcal{U}, \label{eq:slack_input_constraint}\\
& x_{i|k} \in \overline{\mathcal{X}}_i(\xi_{i|k}), \; 0 	\leq \xi_{i|k}, \label{eq:pcbf_state_constraint}\\
& h_f(x_{N|k}) \leq \xi_{N|k}, \; 0 \leq \xi_{N|k}.
\end{align}
\end{subequations}
In~\cite{wabersich2022}, it is shown that the optimal value function of~\eqref{eq:pcbf}, i.e., $\hpb(x)$, is a CBF according to Definition~\ref{def:CBFcont} if $\alpha_f$ is chosen sufficiently large.
\begin{theorem}[\cite{wabersich2022}]\label{thrm:PCBF}
    Let Assumptions~\ref{ass:cont}~and~\ref{ass:terminalCBF} hold. The minimum \eqref{eq:pcbf} exists and if $\alpha_f$ is chosen sufficiently large, then the optimal value function $\hpb(x(k))$ defined in \eqref{eq:cbf_pcbf} is a (predictive) control barrier function according to Definition~\ref{def:CBFcont} with domain $\Dpb := \{ x \in \mathbb{R}^n \mid \hpb(x) \leq \alpha_f \gamma_f \}$ and safe set $\Spb := \{ x \in \mathbb{R}^n \mid \hpb(x) = 0  \}$.
\end{theorem}
Invariance of the safe set $\Spb$ in Theorem~\ref{thrm:PCBF} results from recursive feasibility of the problem with $\xi^*_{i|k}=0$ from standard arguments, while the existence of a continuous, positive decrease $\Delta \hpb$ in $\Dpb\setminus \Spb$ is ensured through the terminal CBF $h_f$ in Assumption~\ref{ass:terminalCBF} if $\xi^*_{N|k}>0$ and through the additional tightening $\Delta_i$ of the state constraints $\overline{\mathcal{X}}_i(\xi_{i|k})$ if $\xi^*_{N|k}=0$. Clearly, if $\xi^*_{i|k}=0$, the problem \eqref{eq:pcbf_sf} is solved with additional tightening on the state constraints and the corresponding tightenings $\Delta_i$ should therefore be chosen small.
An illustration of the different set definitions can be found in~\cite[Figure 4]{wabersich2022}.

\section{EXPLICIT APPROXIMATION OF THE PCBF}~\label{sec:APCBF}
In this section, we present the main proposed approximation-based algorithm. Furthermore, we propose a continuous extension for the PCBF, which is guaranteed to be negative for states in the safe set $\Spb$, i.e., when $\hpb(x)=0$, thereby reducing safety filter interventions caused by approximation errors within $\Spb$.

\subsection{Approximate PCBF} \label{sec:apcbf}
The optimal value function in~\eqref{eq:pcbf} is a CBF according to Definition~\ref{def:CBFclK} and encodes the information of the entire prediction horizon of the predictive optimisation problem with respect to potential constraint violations. However, an explicit formulation of $\hpb$ is generally not available. We therefore propose approximating the PCBF value function with a continuous regressor $\ahpb:\mathbb{R}^n\rightarrow\mathbb{R}$ such that
\begin{equation}
    \ahpb(x)\approx\hpb(x) \;\forall x\in\Dpb.
\end{equation} 
Due to the continuity property of the PCBF in~\eqref{eq:pcbf} and the fact that it can be sampled offline, the function is suitable for approximation, such that a computationally efficient algorithm can be deployed. As the approximation $\ahpb$ will be embedded in an optimisation problem, we consider neural networks in this paper due to the parametric nature of the regressor, allowing for a formulation independent of the number of training samples, compared to, e.g., kernel-based methods. 
Note that while in principle a safety filter policy $\pi(x,u_\mathrm{p})$ could be learned, this policy is generally not only non-continuous and multi-valued, but also has a higher input dimensionality for the data collection and training as $\pi:\mathbb{R}^n\times\mathbb{R}^m\rightarrow\mathbb{R}^m$. As an illustrative example, consider a car approaching an obstacle. If the performance input is given to steer straight towards the obstacle, an optimal solution for the safety filter problem is given by either steering fully to the left or to the right.

In order to train the regressor for the PCBF approximation, a dataset of samples has to be collected offline in the form $\mathbb{D}=\{(x_i, \hpb(x_i))\}_{i=0}^{n_d}$. This data set can be sampled from the entire region of attraction $\Dpb=\{x\in\mathbb{R}^n \mid \hpb(x) \leq \alpha_f\gamma_f \}$, where $\gamma_f$ is the sublevel set of the domain of the terminal CBF $\mathcal{D}_f$ and $\alpha_f$ is the terminal penalty in~\eqref{eq:cbf_pcbf}, or from a subset deemed relevant to the specific application. 
In the Appendix, we propose a geometric sampling method based on \cite{chen2022}, which allows efficiently sampling the PCBF by performing a pseudo random walk over a preselected domain, using previous solutions as warmstarts to the next computed optimisation problem. 

Given a trained explicit approximation $\ahpb$ of $\hpb$, 
the corresponding CBF decrease condition~\eqref{eq:outside_cbf_condition} can directly be used as a constraint similar to CBF-based safety filters in~\cite{agrawal2017}. The optimisation problem which is solved at every time step is then given by
    \begin{align}\label{eq:apcbf_sf}
        \min_u & \Vert u - u_\mathrm{p}(k) \Vert \\
        \textup{s.t. } & \ahpb(f(x(k),u,0)) \!-\! \ahpb(x(k)) \leq - c_\Delta \Delta\ahpb(x(k)), \nonumber \\
        & u\in\mathcal{U}, \nonumber
    \end{align}
where $\Delta\ahpb(x):\mathbb{R}^n\rightarrow\mathbb{R}$ is a decrease function of the approximate PCBF $\ahpb(x)$ and $c_\Delta$ is a user-defined hyperparameter which is 1 if $\Delta \ahpb(x)< 0$ and can be chosen as $c_\Delta\in(0,1]$ otherwise. 
The main advantage of the proposed optimisation problem compared to the PCBF algorithm in~\cite{wabersich2022}, which solves~\eqref{eq:pcbf} and~\eqref{eq:pcbf_sf} consecutively, is that the formulation in~\eqref{eq:apcbf_sf} is independent of the prediction horizon of the predictive control-based optimisation problems. Instead, only a single decrease constraint needs to be satisfied along with the input constraints and the only optimisation variable is $u\in\mathbb{R}^m$. 

Given that the decrease function of $\hpb(x)$, i.e., $\Delta \hpb(x)$, is not explicitly available in a similar fashion to $\hpb(x)$ itself, we propose an optimisation-based formulation for $\Delta\ahpb$. The approximate decrease which is used in~\eqref{eq:apcbf_sf} is defined as
\begin{equation}\label{eq:apcbf_maxdec}
   \Delta \ahpb(x):= \max_{u\in\mathcal{U}} \left(  \ahpb(x) - \ahpb(f(x,u, 0))\right)
\end{equation}
From this formulation of the approximate decrease function, it follows that choosing the hyperparameter $c_\Delta=1$, requires the approximate PCBF $\ahpb$ to undergo the maximal possible decrease. By selecting $c_\Delta\in (0,1)$ when a decrease can be achieved, safety interventions of the proposed filter can be reduced by relaxing the amount that the approximate PCBF is required to decrease. This hyperparameter can be tuned after obtaining $\ahpb$ for desired closed-loop behaviour and provides an additional degree of freedom compared to the the algorithms in \cite{wabersich2022} and \cite{didier2023}, where at every time step the maximal decrease is required to be achieved.

\begin{remark}
    In principle, a decrease function $\Delta \ahpb(x)$ can also be hand-tuned as is common in the continuous-time CBF literature, such as, e.g. choosing an exponential decrease $\Delta \ahpb(x)=\alpha(\hpb(x))$ for some class $\mathcal{K}$ function $\alpha$, see, e.g., \cite{ames2019}. Such a hand-tuned decrease function reduces the computational demand of the proposed scheme by only requiring the solution to one optimisation problem. However, in this case recursive feasibility, which will be shown later in this section for the proposed maximal decrease can no longer be ensured.
\end{remark}

An additional advantage of the explicit approximation of the PCBF is given by the fact that proposed inputs $u_\mathrm{p}(k)$ can be very efficiently certified in terms of satisfying optimality with respect to~\eqref{eq:apcbf_sf} if $u_{\mathrm{p}}(k)\in\mathcal{U}$.
It suffices to perform a simple evaluation, whether $\ahpb(f(x(k),u_\mathrm{p}(k),0))\leq 0$, as it then holds that $u_\mathrm{p}(k)$ is an optimal solution to~\eqref{eq:apcbf_sf}. This function evaluation can thereby reduce the need for solving two optimisation problems online at every time step. Additionally, it can avoid local minima when solving optimisation problems online, especially compared to predictive optimisation problems. The main proposed algorithm is then given in Algorithm~\ref{alg:apcbf}, where first, the proposed input is certified and if the input is not deemed to be safe, the maximal decrease in~\eqref{eq:apcbf_maxdec} is computed and then used in the safety filter~\eqref{eq:apcbf_sf}.

\begin{algorithm}[!b]
\caption{Approximate PCBF safety filter}\label{alg:apcbf}
\begin{algorithmic}[1]
\FOR{$k\geq 0$}
\STATE Measure state $x(k)$
\STATE Obtain the nominal control input $u_\mathrm{p}(k)$
\IF{$u_{\mathrm{p}}(k)\in\mathcal{U} \wedge \ahpb(f(x(k),u_\mathrm{p}(k)),0)\leq 0$}
    \STATE Apply $u(k)\leftarrow u_\mathrm{p}(k)$ \hfill // nominal input is safe
\ELSE 
    \STATE Compute $\Delta \ahpb(x(k))$ using~\eqref{eq:apcbf_maxdec}
    \STATE Solve~\eqref{eq:apcbf_sf} and obtain a solution $u^*$
    \STATE Apply $u(k)\leftarrow u^*$ \hfill // apply filtered input
\ENDIF
\ENDFOR
\end{algorithmic}
\end{algorithm}

\subsection{Continuous Extension of the PCBF} \label{sec:context}
As the PCBF $\hpb(x)=0$ for all $x\in\Spb$, errors in its approximation can lead to undesired conservativeness within $\Spb$ if $\ahpb(x)>0$, as potentially safe control inputs can be deemed unsafe. In order to avoid such unnecessary safety filter interventions, we propose a continuous extension of $\hpb$ in $\Spb$, which is negative. By sampling this continuous extension within $\Spb$ rather than $\hpb$, the data used for training is guaranteed to be negative rather than $0$ and therefore mitigates potential learning errors leading to states being considered unsafe. Consider the function
\begin{equation} \label{eq:PCBFext}
    h_{\mathrm{CE}}(x):=\begin{cases}
    h_\mathrm{PB}(x) & \textup{ if } x\in\Dpb\setminus \Spb, \\
    \bar{h}_{\mathrm{PB}}(x) & \textup{ if } x\in\Spb,
    \end{cases}
\end{equation}
with $h_{\mathrm{CE}}:\Dpb\rightarrow\mathbb{R}$ and where we define 
\begin{subequations}\label{eq:optprobwithnegativeslacks}
\label{eq:pcbfmin}
\begin{align}
\bar{h}_{\mathrm{PB}}(x(k)):=\min_{x_{i|k}, u_{i|k}, \xi_{i|k}} \; \; &\max_{i,j}\{[\xi_{i|k}]_j,\xi_{N|k}\} \label{eq:cbf_pcbf_min} \\
\text{s.t.}\; \; &\forall\, i = 0,\dots,N-1, \nonumber \\
& x_{0|k} = x(k), \label{eq:slack_init_pred_min} \\
& x_{i+1|k} = f(x_{i|k}, u_{i|k}, 0), \\
& u_{i|k} \in \mathcal{U}, \label{eq:slack_input_constraint_min}\\
& x_{i|k} \in \overline{\mathcal{X}}_i(\xi_{i|k}),\label{eq:pcbf_state_constraint_min}\\
& h_f(x_{N|k}) \leq \xi_{N|k}.\label{eq:pcbf_state_constraint_min_term}
\end{align}
\end{subequations}
Note that the slacks in~\eqref{eq:pcbf_state_constraint_min} and~\eqref{eq:pcbf_state_constraint_min_term} are not required to be positive, implying negativity of $\bar{h}_{\mathrm{PB}}$ when the state and terminal constraints are satisfied. This formulation can be solved using an epigraph transformation and shows close resemblance to discrete-time Hamilton-Jacobi value functions such as~\cite{fisac2019}. We show that the function $h_{\mathrm{CE}}$ is continuous on $\Dpb$ and a CBF according to Definition~\ref{def:CBFcont}.
\begin{theorem}\label{thm:continuity}
It holds that $h_{\mathrm{CE}}$ as defined in~\eqref{eq:PCBFext} is a CBF according to Definition~\ref{def:CBFcont}.
\end{theorem}

\begin{proof}
We first prove continuity of $h_{\mathrm{CE}}$ through continuity of $\hpb$ on $\Dpb\setminus\interior\Spb$, continuity of $\bar{h}_{\mathrm{PB}}$ on $\Spb$ and that $\hpb(x)=\bar{h}_\mathrm{PB}(x)=0$ for all $x\in\partial\Spb$.

\textit{1) Continuity of $\hpb$ for all $x\in\Dpb\setminus\interior\Spb$:} In~\cite[Theorem III.6 b)]{wabersich2022}, it is shown that $\hpb$ is continuous for all $x\in\Dpb$, which implies continuity of $h$ for all $x\in\Dpb\setminus\Spb$.

\textit{2) Continuity of $\bar{h}_{\mathrm{PB}}$ for all $x\in\Spb$:} We can show continuity of $\bar{h}_{\mathrm{PB}}$ by using similar arguments as in \cite[Theorem III.6 b)]{wabersich2022}. We show that for any $\epsilon>0$, there exists $\delta>0$ such that $\Vert x-\bar{x}\Vert < \delta$ implies $|\bar{h}_{\mathrm{PB}}(x)-\bar{h}_{\mathrm{PB}}(\bar{x})|<\epsilon$. Consider the optimal solution of \eqref{eq:optprobwithnegativeslacks} at $x$, i.e. $u_{i|k}^*(x)$ as a candidate solution in $\bar{x}$, i.e. $\bar{u}_{i|k}(\bar{x})=u^*_{i|k}$, which results in the state predictions $\bar{x}_{i|k}(\bar{x})$ with $\bar{x}_{0|k}=\bar{x}$. The resulting slacks are then given by $\bar{\xi}_{i|k}(\bar{x})=\min\{0,c(\bar{x}_{i|k})+\Delta_i\mathbf{1}\}$ and $\bar{\xi}_{N|k}=\min\{0,h_f(\bar{x}_{N|k}(\bar{x}))\}$. As the initial state and the corresponding inputs lie in a compact set, i.e. $\bar{x}\in\mathcal{S}_\mathrm{PB}$ and $\bar{u}_{i|k}(\bar{x})\in\mathcal{U}$, and the functions $f,c_j$ and $h_f$ are continuous, it holds that the candidate solution also lies within a compact set from \cite[Lemma C.3]{wabersich2022} and from the Heine-Cantor theorem, it follows that $f,c_j$ and $h_f$ are uniformly continuous. 
As compositions of functions, the minimum and additions of uniform continuous functions all preserve uniform continuity, the state predictions, the candidate slack variables and the corresponding constructed candidate value function $\tilde{h}_\mathrm{PB}(\bar{x})=\max_{i,j}\{[\bar{\xi}_{i|k}(\bar{x})]_j,\bar{\xi}_{N|k}(\bar{x})\}$ are uniformly continuous in the initial condition $\bar{x}$. From the fact that the optimal solution $\bar{h}_{\mathrm{PB}}(\bar{x})\leq \tilde{h}_{\mathrm{PB}}(\bar{x})$ due to the minimisation in \eqref{eq:optprobwithnegativeslacks}, it holds that for any $\epsilon>0$, there exists a uniform $\delta>0$ with $\Vert x-\bar{x}\Vert <\delta$, such that $\bar{h}_{\mathrm{PB}}(\bar{x})-\bar{h}_{\mathrm{PB}}(x)\leq \tilde{h}_{\mathrm{PB}}(\bar{x})-\bar{h}_{\mathrm{PB}}(x)<\epsilon$. The same arguments hold when using the optimal solution at $\bar{x}$ in $x$, i.e. $\bar{u}_{i|k}(x)=u_{i|k}^*(\bar{x})$, such that there exists $\delta>0$ with $\Vert x-\bar{x} \Vert <\delta$, which implies that $\bar{h}_{\mathrm{PB}}(x)-\bar{h}_{\mathrm{PB}}(\bar{x})\leq \tilde{h}_{\mathrm{PB}}(x)-\bar{h}_{\mathrm{PB}}(\bar{x})<\epsilon$, which in turn implies that for any $\epsilon>0$, there exists $\delta>0$ with $\Vert x-\bar{x} \Vert <\delta$, which implies that $|\bar{h}_{\mathrm{PB}}(x)-\bar{h}_{\mathrm{PB}}(\bar{x})|<\epsilon$. 

\textit{3) Equality $\bar{h}_{\mathrm{PB}}(x)=\hpb(x)$ for all $x\in\partial\Spb$}: By definition, we have $h_{\mathrm{PB}}(x)=0$ for all $x\in\partial\Spb$. 
In order to show that $\bar{h}_{\mathrm{PB}}(x)=0$ for all $x\in\partial\Spb$ we first consider an optimal input sequence $\tilde{u}_{i|k}^*(\tilde{x})$ of~\eqref{eq:optprobwithnegativeslacks} for any $\tilde{x}\in\partial\Spb$. Given this optimal input sequence, it holds that the corresponding slack variables $\tilde{\xi}_{i|k}(\tilde{x})$ and $\tilde{\xi}_{N|k}(\tilde{x})$ and the optimal value function $\bar{h}_{\mathrm{PB}}(\tilde{x})$ are continuous in $\tilde{x}$ due to continuity of $f$ and $c_x$.
Furthermore, it holds that the candidate slacks $\tilde{\xi}_{i|k}(\tilde{x})\leq 0$ and $\tilde{\xi}_{N|k}(\tilde{x})\leq 0$ as any optimal inputs $u^*_{i|k}(\tilde{x})$ of~\eqref{eq:pcbf} are feasible in~\eqref{eq:optprobwithnegativeslacks} and $\hpb(x)=0$ for all $x\in\partial\Spb\subseteq\Spb$. As the candidate slacks are all non-positive, it follows that $\bar{h}_{\mathrm{PB}}(x)\leq 0$ for all $x\in\partial\Spb$. We can then show that $\bar{h}_{\mathrm{PB}}(x)= 0$ for all $x\in\partial\Spb\subseteq\Spb$ by contradiction. Suppose that for any $\tilde{x}\in\partial\Spb$, it holds that there exists a solution with all slack variables $\xi_{i|k}(\tilde{x})<0$ and $\xi_{N|k}(\tilde{x})<0$ in~\eqref{eq:optprobwithnegativeslacks}. Then for the corresponding optimal input sequence $\tilde{u}_{i|k}(\tilde{x})$, it holds that $\xi_{i|k}(\tilde{x})$ and $\xi_{N|k}(\tilde{x})$ are continuous in $\tilde{x}$. The continuity of the slack variables in $\tilde{x}$ implies that there exists an open ball of radius $\delta>0$, i.e. $B_\mathcal{D}^\delta(\tilde{x})=\{x\in\mathcal{D}\mid\Vert x-\tilde{x}\Vert < \delta\}$, for which it still holds that $\xi_{i|k}(\tilde{x})<0$ and $\xi_{N|k}(\tilde{x})<0$. As the optimal input $\tilde{u}_{i|k}(\tilde{x})$ is a feasible input for~\eqref{eq:pcbf}, it implies that $\hpb(x)=0$ for all $x\in B_\mathcal{D}^\delta(\tilde{x})$, which implies that $B_\mathcal{D}^\delta(\tilde{x})\subseteq\Spb$. This fact however contradicts with the initial assumption that $\tilde{x}\in\partial\Spb$, as by definition any open neighbourhood of $\tilde{x}\in\partial\Spb$ must contain at least one point in $\Dpb\setminus\Spb$. By contradiction, it therefore holds that any optimal solution of~\eqref{eq:optprobwithnegativeslacks} has at least one slack $[\tilde{\xi}^*_{i|k}(\tilde{x})]_j$ or $\tilde{\xi}^*_{N|k}(\tilde{x})$ which is 0 and all slacks are non-positive, which implies that $\bar{h}_{\mathrm{PB}}(\tilde{x})=\max_{i,j}\{[\tilde{\xi}^*_{i|k}(\tilde{x})]_j, \tilde{\xi}^*_{N|k}(\tilde{x})\} =0$.

As the function $h_{\mathrm{CE}}$ is continuous on its subdomains $\Spb$ and $\Dpb\setminus\Spb$ and $\bar{h}_{\mathrm{PB}}(x)=\hpb(x)=0$ for all $x\in\partial\Spb$, it holds that $h_{\mathrm{CE}}$ is continuous on $\Dpb$. As $h(x)=\hpb(x)$ for all $x\in\Dpb\setminus \Spb$ and $h_{\mathrm{CE}}(x)\leq 0$ for all $x\in\Spb$, the conditions in Definition~\ref{def:CBFcont} are satisfied.
\end{proof}

\begin{remark}
Note that the function
\begin{equation}
h(x):=\begin{cases}
    h_\mathrm{PB}(x) & \textup{ if } x\in\Dpb\setminus \Spb, \\
    c_h\bar{h}_{\mathrm{PB}}(x) & \textup{ if } x\in\Spb.
    \end{cases}
\end{equation}
 is also a CBF according to Definition~\ref{def:CBFcont} for any constant $c_h\geq 0$ as continuity of $h$ is preserved. The design parameter $c_h$ allows for trading off a reduced impact of learning errors in $\Spb$ on the safety filter interventions and ease of approximation.
\end{remark}

In principle, this continuous extension requires solving two optimisation problems for each sample $x\in\Spb$, i.e., \eqref{eq:pcbf} to determine if $x\in\Spb$ or $x\in\Dpb\setminus \Spb$ and \eqref{eq:optprobwithnegativeslacks} in case $x\in\Spb$. However, a geometric sampler such as used in \cite{chen2022} can be used, which performs a line search over given seed and target points. By using such a geometric sampler, information of the previous solution can be used to solve either \eqref{eq:pcbf} or \eqref{eq:optprobwithnegativeslacks} based on which problem was previously solved. It follows that two optimisation problems only need to be solved if the samples on the line cross over from $\Spb$ into $\Dpb\setminus\Spb$ or vice-versa. Additionally, this geometric sampler benefits from being able to warmstart the solver using the previous optimal solution $\mathbf{o}=\{x^*_{i}, u^*_{i}, \xi^*_{i}, \bm{\nu}^*, \bm{\lambda}^*\}$, where $\bm{\nu}^*$ and $\bm{\lambda}^*$ denote the optimal dual variables. Given that the optimal solution is shown to be continuous in $x$ in \cite[Theorem III.6]{wabersich2022} and Theorem~\ref{thm:continuity}, the previously computed optimal solution will be close to the solution of the next iteration in the line search. This allows the computation time during sampling to be significantly reduced as shown in \cite{chen2022}. Finally, the random walk can be parallelised by performing line search for different goal points given a random starting seed. The resulting sampling algorithm is then given by Algorithm~\ref{alg:sampling} in the Appendix.

\section{MAIN THEORETICAL ANALYSIS}\label{sec:Theory}
In this section, we provide the main theoretical analysis that provides the foundation of the approximation algorithm in Section~\ref{sec:apcbf}. We first show that every CBF according to Definition~\ref{def:CBFcont} implies existence of a CBF according to Definition~\ref{def:CBFclK}. This result enables analysis of the closed-loop behaviour of the predictive CBF algorithm in the presence of approximation errors as well as exogenous disturbances.

\subsection{Relation of CBF Definitions}
In order to show the existence of a CBF according to Definition~\ref{def:CBFclK} given a CBF according to Definition~\ref{def:CBFcont}, we need to relate the continuous function $\hpb$ and its corresponding decrease $\Delta \hpb$ to class $\mathcal{K}$ comparison functions.
We first recall the class $\mathcal{K}$ lower bound  proposed in~\cite{didier2023}.

\begin{lemma}[reformulated from~\cite{didier2023}]
\label{lemma:clK_bound}~ \\
Consider the compact and non-empty sets $\mathcal{S}$ and $\mathcal{D}$, with $\mathcal{S}\subset\mathcal{D}$ and a continuous function $\gamma:\mathcal{D}\setminus\mathcal{S}\rightarrow \mathbb{R}$, with $\gamma(x)>0$ for all $x\in\mathcal{D}\setminus\mathcal{S}$. Then there exists a $\mathcal{K}$-function $\alpha:[0,\bar{r}] \rightarrow \mathbb{R}$ with $\bar{r}:=\max_{x\in\mathcal{D}} \Vert x \Vert_\mathcal{S}$ such that
\begin{equation}
    \gamma(x)\geq \alpha(\Vert x \Vert_\mathcal{S}) \; \forall x\in\mathcal{D}\setminus \mathcal{S}.
\end{equation}
\end{lemma}

\begin{proof}
Define the auxiliary function $p\!:\! [0,\overline{r}] {\to} \mathbb{R}_{\geq0}$ as
\begin{align}
\label{eq:prop_1_min_bound}
    \hspace{-0.1cm}p(r) \!:=\! 
    \begin{cases}
    0 &\hspace{-0.5cm} \text{if } r = 0, \\
    & \\
    \parbox{4.7cm}{\footnotesize $\min\limits_x \;\; \gamma(x) \\
    \text{s.t.} \;\;  x \in \mathcal{D}\setminus \{x \in \mathcal{D} \mid \Vert x\Vert_{\mathcal{S}} < r\}$}& \hspace{-0.5cm}\mathrm{otherwise}.
    \end{cases}
\end{align}
Since $\gamma$ is a continuous function and $\mathcal{D} \setminus \{x \in \mathcal{D} \mid \Vert x\Vert_{\mathcal{S}} < r\}$ is a compact set, the above minimum is attained for all $r\in(0,\overline{r}]$. 
It holds that
\begin{align*}
p(\Vert x \Vert_{\mathcal{S}}) \leq  \gamma(x), \; \forall x \in \mathcal{D}\setminus\mathcal{S},
\end{align*}
which follows directly from the above minimization as $x$ is a feasible point if $\Vert x \Vert_{\mathcal{S}} > 0$. Additionally, it holds that $p$ is non-decreasing with increasing $r$, since $\tilde r > r$ implies $\mathcal{D}\setminus\{x \in \mathcal{D} \mid \Vert x\Vert_{\mathcal{S}} < \tilde r\}\subseteq\mathcal{D}\setminus\{x \in \mathcal{D} \mid \Vert x\Vert_{\mathcal{S}} < r\}$, i.e., the domain in the constrained minimization problem (second case in \eqref{eq:prop_1_min_bound}) is non-increasing.
As $p(r)$ is not necessarily continuous for non-convex sets $\mathcal S$, we construct a class $\mathcal{K}$ function $\alpha : [0,\overline{r}] \to \mathbb{R}_{\geq0}$ which is a lower bound of $p$ on $\mathcal{D}\setminus\mathcal{S}$ similarly to the result in~\cite[Lemma 2.5]{clarke98}.
Consider the partitions $(R_{k+1},R_k]$ on the interval $(0,\overline{r}]$ with boundary points $R_k := \overline{r}2^{-k} , \mathrm{for} \; k = 0, 1, 2, \dots$ and define $\phi_k:=2^{-k} p(R_{k+1})$.
Note that $\phi_{k+1} < \phi_{k}$, since $p(r)$ is non-decreasing and $2^{-k-1} < 2^{-k}$ $\forall k \geq 0$. In addition, it holds that $\phi_k\leq p(r), \forall r \in (R_{k+1},R_k]$ as $\phi_k=2^{-k} p(R_{k+1}) \leq p(R_{k+1})$.

We construct the function
\begin{align*}
\alpha(r) = \begin{cases}
    0 & \text{if } r=0, \\
    \phi_{k+1} + \frac{\phi_{k} - \phi_{k+1}}{(R_{k}-R_{k+1})}(r-R_{k+1}),& r\in(R_{k+1},R_k],
\end{cases}
\end{align*}
which is piece-wise affine on successive intervals $r \in (R_{k+1}, R_{k}]$, continuous, as $\lim_{k\rightarrow \infty} \phi_k = 0$, strictly increasing and is therefore a class $\mathcal{K}$ function. Furthermore, it holds that $p(\Vert x\Vert_\mathcal{S}) \leq \gamma(x)$, which implies $\alpha(\Vert x \Vert_{\mathcal{S}}) \leq \gamma(x) \; \forall x\in\mathcal{D}\setminus\mathcal{S}$ as desired.
\end{proof}
In a similar fashion, an upper bound for positive continuous functions can be constructed, which to the best of the author's knowledge has not been previously proposed in this form. Note that in~\cite[Proposition B.25]{rawlings2017}, such an upper bound is shown only with respect to a singleton $\mathcal{S}=\{0\}$.

\begin{lemma}
\label{lemma:clK_upbound}
Consider the compact and non-empty sets $\mathcal{S}$ and $\mathcal{D}$, with $\mathcal{S}\subset\mathcal{D}$ and a continuous function $\gamma:\mathcal{D}\setminus \interior\mathcal{S}\rightarrow\mathbb{R}$, with $\gamma(x)>0$ for all $x\in\mathcal{D}\setminus\mathcal{S}$ and $\gamma(x)=0$ for all $x\in\partial\mathcal{S}$. Then there exists a $\mathcal{K}$-function $\beta:[0,\bar{r}]\rightarrow\mathbb{R}_{\geq 0}$, with $\bar{r}:=\max_{x\in\mathcal{D}}\Vert x \Vert_\mathcal{S}$ such that
\begin{equation}
    \gamma(x)\leq \beta(\Vert x \Vert_\mathcal{S}) \; \forall x\in\mathcal{D}\setminus \interior\mathcal{S}.
\end{equation}
\end{lemma}

\begin{proof}
Define the auxiliary function $q\!:\![0,\bar{r}]\rightarrow\mathbb{R}_{\geq0}$ as 
\begin{align} \label{eq:upperbound_aux}
q(r) :=  \max\limits_x \;\; & \gamma(x) \\
    \text{s.t.}  \; & x \in (\mathcal{D}\setminus\interior\mathcal{S}) \setminus\{x \in \mathcal{D} \mid \Vert x\Vert_{\mathcal{S}} > r\}. \nonumber
\end{align}
As $\gamma$ is a continuous function and $(\mathcal{D}\setminus\interior\mathcal{S}) \setminus\{x \in \mathcal{D} \mid \Vert x\Vert_{\mathcal{S}} > r\}$ is a compact set, the above maximum is attained for all $r \in [0,\bar{r}]$. It holds that 
$$q(\Vert x \Vert_\mathcal{S})\geq \gamma(x) \;\;\forall x\in \mathcal{D}\setminus\interior\mathcal{S}$$
which follows directly from the fact that $x$ is a feasible point in the above maximisation and the fact that $\gamma(x) = 0, \; \forall x \in \partial\mathcal{S}$. Additionally, it holds that $q$ is non-decreasing with $r$, since $\tilde{r}> r$ implies $(\mathcal{D}\setminus\interior\mathcal{S})\setminus\{x \in \mathcal{D} \mid \Vert x\Vert_{\mathcal{S}} > r\}\subseteq(\mathcal{D}\setminus\interior\mathcal{S})\setminus\{x \in \mathcal{D} \mid \Vert x\Vert_{\mathcal{S}} > \tilde{r}\}$, i.e., the domain in the constrained maximisation problem in \eqref{eq:upperbound_aux} is non-decreasing.

A continuous, strictly increasing upper bound for $q(r)$ is then given in the proof of \cite[Lemma 2.5]{clarke98} by defining the partitions $[R_{k+1},R_k)$ with $R_k=2^{-k}\bar{r}$ for $k\geq 0$, the quantities $\phi_k=q(R_{k-1})+2^{-k}$ for $k \geq 0$ with $R_{-1}=\bar{r}$ and the piece-wise affine function
\begin{equation*}
\beta(r)=\begin{cases}
 0 & \textup{ if } r=0\\
 \phi_{k+1} + \frac{\phi_k-\phi_{k+1}}{R_k-R_{k+1}}(r-R_{k+1}) & \textup{ if } r\in[R_{k+1}, R_k) \\
 \phi_0 & \textup{ if } r=\bar{r}.
\end{cases}
\end{equation*}

It then holds that $\beta(0)=0$, $\beta$ is strictly increasing as $\phi_{k}\geq\phi_{k+1}$ as $2^{-k}>2^{-k-1}$ and $q$ is increasing, and $\beta$ is continuous as $\lim_{r\rightarrow 0^+}\beta(r)=0$ as $\lim_{k\rightarrow \infty }\phi_k=0$ and $\lim_{r\rightarrow\bar{r}^-}\beta(r)=\beta(\bar{r})$, therefore we have $\beta\in\mathcal{K}$.
\end{proof}

\begin{figure}[!t]
    \vspace{0.3cm}
    \centering
    \tikzset{every picture/.style={line width=0.75pt}} %set default line width to 0.75pt        
    
    \begin{tikzpicture}[x=0.75pt,y=0.75pt,yscale=-1,xscale=1]
    %uncomment if require: \path (0,447); %set diagram left start at 0, and has height of 447
    
    %Shape: Rectangle [id:dp775220524432878] 
    \draw  [color={rgb, 255:red, 74; green, 144; blue, 226 }  ,draw opacity=1 ][fill={rgb, 255:red, 74; green, 144; blue, 226 }  ,fill opacity=1 ] (109.48,140.81) -- (391.12,140.81) -- (391.12,144) -- (109.48,144) -- cycle ;
    %Shape: Rectangle [id:dp06944854436067227] 
    \draw  [color={rgb, 255:red, 208; green, 2; blue, 27 }  ,draw opacity=1 ][fill={rgb, 255:red, 208; green, 2; blue, 27 }  ,fill opacity=1 ] (147.81,137.74) -- (190.55,137.74) -- (190.55,139.74) -- (147.81,139.74) -- cycle ;
    %Shape: Rectangle [id:dp5175467970000323] 
    \draw  [color={rgb, 255:red, 208; green, 2; blue, 27 }  ,draw opacity=1 ][fill={rgb, 255:red, 208; green, 2; blue, 27 }  ,fill opacity=1 ] (233.66,137.36) -- (276.38,137.36) -- (276.38,140) -- (233.66,140) -- cycle ;
    %Shape: Rectangle [id:dp05514705474744219] 
    \draw  [color={rgb, 255:red, 208; green, 2; blue, 27 }  ,draw opacity=1 ][fill={rgb, 255:red, 208; green, 2; blue, 27 }  ,fill opacity=1 ] (305.47,137.86) -- (347.73,137.86) -- (347.73,140) -- (305.47,140) -- cycle ;
    \draw  [dash pattern={on 4.5pt off 4.5pt}]  (233.66,137.36) -- (233.66,70) ;
    \draw  [dash pattern={on 4.5pt off 4.5pt}]  (276.38,137.36) -- (276.38,70) ;
    \draw  [dash pattern={on 4.5pt off 4.5pt}]  (147.81,137.36) -- (147.81,70) ;
    \draw  [dash pattern={on 4.5pt off 4.5pt}]  (190.55,137.36) -- (190.55,70) ;
    \draw  [dash pattern={on 4.5pt off 4.5pt}]  (305.47,137.36) -- (305.47,70) ;
    \draw  [dash pattern={on 4.5pt off 4.5pt}]  (347.73,137.36) -- (347.73,70) ;
    %Shape: Axis 2D [id:dp7712336074756354] 
    \draw  (90.48,144) -- (407.09,144)(109.48,54) -- (109.48,154) (400.09,139) -- (407.09,144) -- (400.09,149) (104.48,61) -- (109.48,54) -- (114.48,61)  ;
    %Curve Lines [id:da840193942531426] 
    \draw [line width=1.5]    (109.48,75) .. controls (133.73,92.92) and (143.4,96.67) .. (147.81,139.74) ;
    %Curve Lines [id:da01957049021110091] 
    \draw [line width=1.5]    (276.38,140) .. controls (287.84,125) and (282.01,101.74) .. (291.32,101.74) ;
    %Curve Lines [id:da330021912481947] 
    \draw [line width=1.5]    (291.32,101.74) .. controls (299.2,101.74) and (295.8,139.86) .. (305.47,140) ;
    %Curve Lines [id:da865421065013299] 
    \draw [line width=1.5]    (347.73,140) .. controls (352.93,85.17) and (364.93,75.42) .. (391.07,70.36) ;
    %Shape: Axis 2D [id:dp2069341197361343] 
    \draw  (90.48,248) -- (407.09,248)(109.48,158) -- (109.48,258) (400.09,243) -- (407.09,248) -- (400.09,253) (104.48,165) -- (109.48,158) -- (114.48,165)  ;
    %Shape: Axis 2D [id:dp13818788731088394] 
    \draw  (90.48,352) -- (407.09,352)(109.48,262) -- (109.48,362) (400.09,347) -- (407.09,352) -- (400.09,357) (104.48,269) -- (109.48,262) -- (114.48,269)  ;
    %Curve Lines [id:da11342657156854008] 
    \draw [line width=1.5]    (190.55,139.74) .. controls (195.69,107.92) and (198.2,79.42) .. (202.97,97.07) ;
    %Curve Lines [id:da651061906941812] 
    \draw [line width=1.5]    (202.97,97.07) .. controls (212.53,131.92) and (212.89,94.42) .. (214.68,92.17) ;
    %Curve Lines [id:da2706295411221744] 
    \draw [line width=1.5]    (214.68,92.17) .. controls (222.56,67.17) and (216.12,137.41) .. (233.66,140) ;
    %Straight Lines [id:da4526241596739875] 
    \draw [color={rgb, 255:red, 65; green, 117; blue, 5 }  ,draw opacity=1 ][line width=1.5]    (109.56,70.11) -- (129.2,70.32) ;
    \draw [shift={(133.2,70.36)}, rotate = 180.61] [fill={rgb, 255:red, 65; green, 117; blue, 5 }  ,fill opacity=1 ][line width=0.08]  [draw opacity=0] (4.64,-2.23) -- (0,0) -- (4.64,2.23) -- cycle    ;
    %Straight Lines [id:da12569059088641654] 
    \draw [color={rgb, 255:red, 65; green, 117; blue, 5 }  ,draw opacity=1 ][line width=1.5]    (208.83,70.36) -- (215.33,70.36) ;
    \draw [shift={(219.33,70.36)}, rotate = 180] [fill={rgb, 255:red, 65; green, 117; blue, 5 }  ,fill opacity=1 ][line width=0.08]  [draw opacity=0] (4.64,-2.23) -- (0,0) -- (4.64,2.23) -- cycle    ;
    \draw [shift={(204.83,70.36)}, rotate = 0] [fill={rgb, 255:red, 65; green, 117; blue, 5 }  ,fill opacity=1 ][line width=0.08]  [draw opacity=0] (4.64,-2.23) -- (0,0) -- (4.64,2.23) -- cycle    ;
    %Straight Lines [id:da15413721285442827] 
    \draw [color={rgb, 255:red, 65; green, 117; blue, 5 }  ,draw opacity=1 ][line width=1.5]    (366.24,70.36) -- (391.07,70.36) ;
    \draw [shift={(362.24,70.36)}, rotate = 0] [fill={rgb, 255:red, 65; green, 117; blue, 5 }  ,fill opacity=1 ][line width=0.08]  [draw opacity=0] (4.64,-2.23) -- (0,0) -- (4.64,2.23) -- cycle    ;
    %Straight Lines [id:da5611451370752778] 
    \draw [color={rgb, 255:red, 99  ; green, 99; blue, 99}  ,draw opacity=1 ][fill={rgb, 255:red, 99  ; green, 99; blue, 99}  ,fill opacity=1 ][line width=1.5]    (147.52,70.36) -- (137.2,70.36) ;
    \draw [shift={(133.2,70.36)}, rotate = 360] [fill={rgb, 255:red, 99  ; green, 99; blue, 99}  ,fill opacity=1 ][line width=0.08]  [draw opacity=0] (4.64,-2.23) -- (0,0) -- (4.64,2.23) -- cycle    ;
    %Straight Lines [id:da4610672942563101] 
    \draw [color={rgb, 255:red, 99  ; green, 99; blue, 99}  ,draw opacity=1 ][fill={rgb, 255:red, 99  ; green, 99; blue, 99}  ,fill opacity=1 ][line width=1.5]    (190.86,70.11) -- (200.83,70.29) ;
    \draw [shift={(204.83,70.36)}, rotate = 181.03] [fill={rgb, 255:red, 99  ; green, 99; blue, 99}  ,fill opacity=1 ][line width=0.08]  [draw opacity=0] (4.64,-2.23) -- (0,0) -- (4.64,2.23) -- cycle    ;
    %Straight Lines [id:da6982662474657138] 
    \draw [color={rgb, 255:red, 99  ; green, 99; blue, 99}  ,draw opacity=1 ][fill={rgb, 255:red, 99  ; green, 99; blue, 99}  ,fill opacity=1 ][line width=1.5]    (223.33,70.43) -- (233.3,70.61) ;
    \draw [shift={(219.33,70.36)}, rotate = 1.03] [fill={rgb, 255:red, 99  ; green, 99; blue, 99}  ,fill opacity=1 ][line width=0.08]  [draw opacity=0] (4.64,-2.23) -- (0,0) -- (4.64,2.23) -- cycle    ;
    %Straight Lines [id:da9388598425749695] 
    \draw [color={rgb, 255:red, 99  ; green, 99; blue, 99}  ,draw opacity=1 ][fill={rgb, 255:red, 99  ; green, 99; blue, 99}  ,fill opacity=1 ][line width=1.5]    (294.79,70.43) -- (304.75,70.61) ;
    \draw [shift={(290.79,70.36)}, rotate = 1.03] [fill={rgb, 255:red, 99  ; green, 99; blue, 99}  ,fill opacity=1 ][line width=0.08]  [draw opacity=0] (4.64,-2.23) -- (0,0) -- (4.64,2.23) -- cycle    ;
    %Straight Lines [id:da32474889142543595] 
    \draw [color={rgb, 255:red, 99  ; green, 99; blue, 99}  ,draw opacity=1 ][fill={rgb, 255:red, 99  ; green, 99; blue, 99}  ,fill opacity=1 ][line width=1.5]    (276.82,70.11) -- (286.79,70.29) ;
    \draw [shift={(290.79,70.36)}, rotate = 181.03] [fill={rgb, 255:red, 99  ; green, 99; blue, 99}  ,fill opacity=1 ][line width=0.08]  [draw opacity=0] (4.64,-2.23) -- (0,0) -- (4.64,2.23) -- cycle    ;
    %Straight Lines [id:da2914557620865277] 
    \draw [color={rgb, 255:red, 99  ; green, 99; blue, 99}  ,draw opacity=1 ][fill={rgb, 255:red, 99  ; green, 99; blue, 99}  ,fill opacity=1 ][line width=1.5]    (348.27,70.11) -- (358.24,70.29) ;
    \draw [shift={(362.24,70.36)}, rotate = 181.03] [fill={rgb, 255:red, 99  ; green, 99; blue, 99}  ,fill opacity=1 ][line width=0.08]  [draw opacity=0] (4.64,-2.23) -- (0,0) -- (4.64,2.23) -- cycle    ;
    %Curve Lines [id:da9657875784165451] 
    \draw [color={rgb, 255:red, 65; green, 117; blue, 5 }  ,draw opacity=1 ][line width=1.5]    (109.48,248) .. controls (153.07,245.17) and (141.95,224.83) .. (190.84,220.33) ;
    %Flowchart: Connector [id:dp12099588306301823] 
    \draw  [color={rgb, 255:red, 65; green, 117; blue, 5 }  ,draw opacity=1 ][fill={rgb, 255:red, 65; green, 117; blue, 5 }  ,fill opacity=1 ] (232.48,220.58) .. controls (232.48,219.96) and (232.84,219.46) .. (233.28,219.46) .. controls (233.73,219.46) and (234.09,219.96) .. (234.09,220.58) .. controls (234.09,221.2) and (233.73,221.71) .. (233.28,221.71) .. controls (232.84,221.71) and (232.48,221.2) .. (232.48,220.58) -- cycle ;
    %Flowchart: Connector [id:dp5239203756988067] 
    \draw  [color={rgb, 255:red, 65; green, 117; blue, 5 }  ,draw opacity=1 ] (232.5,201.17) .. controls (232.5,200.55) and (232.5,200.04) .. (233.84,200.04) .. controls (235,200.04) and (235,200.55) .. (235,201.17) .. controls (235,201.79) and (235,202.29) .. (233.84,202.29) .. controls (232.5,202.29) and (232.5,201.79) .. (232.5,201.17) -- cycle ;
    %Curve Lines [id:da6211198893929049] 
    \draw [color={rgb, 255:red, 65; green, 117; blue, 5 }  ,draw opacity=1 ][line width=1.5]    (190.12,220.33) .. controls (193.53,220.33) and (230.06,220.58) .. (233.28,220.58) ;
    %Curve Lines [id:da3829392932250215] 
    \draw [color={rgb, 255:red, 65; green, 117; blue, 5 }  ,draw opacity=1 ][line width=1.5]    (234.64,200.17) .. controls (268.13,171.17) and (337.35,171.17) .. (398.06,170.42) ;
    %Straight Lines [id:da2773210100748533] 
    \draw    (398.39,245.06) -- (398.39,251.06) ;
    %Straight Lines [id:da42810798376787207] 
    \draw    (255.13,245.06) -- (255.13,251.06) ;
    %Straight Lines [id:da6544445618671346] 
    \draw    (183.5,245.06) -- (183.5,251.06) ;
    %Straight Lines [id:da5578452048149283] 
    \draw  [dash pattern={on 4.5pt off 4.5pt}]  (254.95,187.83) -- (254.95,247.83) ;
    %Straight Lines [id:da5986238288557286] 
    \draw  [dash pattern={on 4.5pt off 4.5pt}]  (254.95,187.83) -- (397.68,187.33) ;
    %Flowchart: Connector [id:dp14698387967182036] 
    \draw  [fill={rgb, 255:red, 0; green, 0; blue, 0 }  ,fill opacity=1 ] (396.87,187.33) .. controls (396.87,186.71) and (397.23,186.21) .. (397.68,186.21) .. controls (398.12,186.21) and (398.48,186.71) .. (398.48,187.33) .. controls (398.48,187.95) and (398.12,188.46) .. (397.68,188.46) .. controls (397.23,188.46) and (396.87,187.95) .. (396.87,187.33) -- cycle ;
    %Straight Lines [id:da2781941550744276] 
    \draw  [dash pattern={on 4.5pt off 4.5pt}]  (183.32,221.83) -- (183.5,248.58) ;
    %Straight Lines [id:da603272979745969] 
    \draw  [dash pattern={on 4.5pt off 4.5pt}]  (183.32,221.83) -- (254.76,222.04) ;
    %Flowchart: Connector [id:dp6086619549855186] 
    \draw  [fill={rgb, 255:red, 0; green, 0; blue, 0 }  ,fill opacity=1 ] (253.97,234.58) .. controls (253.97,233.96) and (254.33,233.46) .. (254.77,233.46) .. controls (255.22,233.46) and (255.58,233.96) .. (255.58,234.58) .. controls (255.58,235.2) and (255.22,235.71) .. (254.77,235.71) .. controls (254.33,235.71) and (253.97,235.2) .. (253.97,234.58) -- cycle ;
    %Flowchart: Connector [id:dp8138773243394852] 
    \draw  [fill={rgb, 255:red, 0; green, 0; blue, 0 }  ,fill opacity=1 ] (182.69,245.06) .. controls (182.69,244.43) and (183.05,243.93) .. (183.5,243.93) .. controls (183.94,243.93) and (184.3,244.43) .. (184.3,245.06) .. controls (184.3,245.68) and (183.94,246.18) .. (183.5,246.18) .. controls (183.05,246.18) and (182.69,245.68) .. (182.69,245.06) -- cycle ;
    %Straight Lines [id:da9053422099607242] 
    \draw [line width=1.5]    (109.48,248) -- (183.5,245.06) ;
    %Straight Lines [id:da5289597035486449] 
    \draw [line width=1.5]    (183.5,245.06) -- (254.77,234.58) ;
    %Straight Lines [id:da49894677624479766] 
    \draw [line width=1.5]    (255.58,234.58) -- (397.68,187.33) ;
    %Curve Lines [id:da20589804605115525] 
    \draw [color={rgb, 255:red, 99  ; green, 99; blue, 99}  ,draw opacity=1 ][line width=1.5]    (109.48,352) .. controls (138.13,322) and (151.92,322.75) .. (166.6,322.25) ;
    %Straight Lines [id:da08312621350666283] 
    \draw [color={rgb, 255:red, 99  ; green, 99; blue, 99}  ,draw opacity=1 ][line width=1.5]    (166.6,322.25) -- (195.26,322.25) ;
    %Curve Lines [id:da8098801954269099] 
    \draw [color={rgb, 255:red, 99  ; green, 99; blue, 99}  ,draw opacity=1 ][line width=1.5]    (195.26,322.25) .. controls (201.17,303.75) and (257.04,292.5) .. (398.15,290.5) ;
    %Straight Lines [id:da02456944141858286] 
    \draw    (398.22,349.31) -- (398.22,355.31) ;
    %Straight Lines [id:da92457525260469] 
    \draw    (254.95,349.31) -- (254.95,355.31) ;
    %Straight Lines [id:da6644012145189049] 
    \draw    (183.32,349.31) -- (183.32,355.31) ;
    %Straight Lines [id:da05209988625341566] 
    \draw  [dash pattern={on 4.5pt off 4.5pt}]  (398.15,290.5) -- (398.15,352.5) ;
    %Straight Lines [id:da04762011885747808] 
    \draw  [dash pattern={on 0.84pt off 2.51pt}]  (398.15,275.5) -- (398.15,290.5) ;
    %Straight Lines [id:da14302069827734165] 
    \draw  [dash pattern={on 4.5pt off 4.5pt}]  (254.89,291) -- (254.71,351.75) ;
    %Straight Lines [id:da9115855782240687] 
    \draw  [dash pattern={on 4.5pt off 4.5pt}]  (254.89,291) -- (398.15,290.5) ;
    %Flowchart: Connector [id:dp06022314321962874] 
    \draw  [fill={rgb, 255:red, 0; green, 0; blue, 0 }  ,fill opacity=1 ] (397.35,275.5) .. controls (397.35,274.88) and (397.71,274.38) .. (398.15,274.38) .. controls (398.6,274.38) and (398.96,274.88) .. (398.96,275.5) .. controls (398.96,276.12) and (398.6,276.63) .. (398.15,276.63) .. controls (397.71,276.63) and (397.35,276.12) .. (397.35,275.5) -- cycle ;
    %Flowchart: Connector [id:dp86442092743927] 
    \draw  [fill={rgb, 255:red, 0; green, 0; blue, 0 }  ,fill opacity=1 ] (254.08,282.63) .. controls (254.08,282) and (254.45,281.5) .. (254.89,281.5) .. controls (255.34,281.5) and (255.7,282) .. (255.7,282.63) .. controls (255.7,283.25) and (255.34,283.75) .. (254.89,283.75) .. controls (254.45,283.75) and (254.08,283.25) .. (254.08,282.63) -- cycle ;
    %Flowchart: Connector [id:dp9857665471342116] 
    \draw  [fill={rgb, 255:red, 0; green, 0; blue, 0 }  ,fill opacity=1 ] (182.51,293.33) .. controls (182.51,292.71) and (182.87,292.21) .. (183.32,292.21) .. controls (183.76,292.21) and (184.12,292.71) .. (184.12,293.33) .. controls (184.12,293.95) and (183.76,294.46) .. (183.32,294.46) .. controls (182.87,294.46) and (182.51,293.95) .. (182.51,293.33) -- cycle ;
    %Straight Lines [id:da24516449054936285] 
    \draw  [dash pattern={on 0.84pt off 2.51pt}]  (254.89,283.75) -- (254.89,291) ;
    %Straight Lines [id:da4678365174948431] 
    \draw  [dash pattern={on 4.5pt off 4.5pt}]  (183.32,299.08) -- (254.77,298.58) ;
    %Straight Lines [id:da25874940153075676] 
    \draw  [dash pattern={on 4.5pt off 4.5pt}]  (183.32,299.08) -- (183.26,352.25) ;
    
    \draw  [dash pattern={on 4.5pt off 4.5pt}]  (153.32,159.08) -- (153.26,352.25) ;
    \draw [line width=0.8]    (153.32,348.25) -- (153.32,355.25) ;

    %Straight Lines [id:da6251812666426524] 
    \draw  [dash pattern={on 0.84pt off 2.51pt}]  (183.26,296.75) -- (183.32,299.08) ;
    %Straight Lines [id:da37624916107121154] 
    \draw [line width=1.5]    (183.32,293.33) -- (254.89,282.63) ;
    %Straight Lines [id:da4468954695643055] 
    \draw [line width=1.5]    (254.89,282.63) -- (398.15,275.5) ;
    %Straight Lines [id:da6367211535536543] 
    \draw [line width=1.5]    (144.93,319.25) -- (183.32,293.33) ;
    %Straight Lines [id:da9840457320612204] 
    \draw [line width=1.5]    (133.47,325) -- (144.93,319.25) ;
    %Straight Lines [id:da5377259404019112] 
    \draw [line width=1.5]    (127.74,330) -- (133.47,325) ;
    %Curve Lines [id:da1412498974887091] 
    \draw [line width=1.5]    (109.48,352) .. controls (114.85,345.5) and (120.94,336.5) .. (127.74,330) ;
    
    % Text Node
    \draw (382.78,119.5) node [anchor=north west][inner sep=0.75pt]  [color={rgb, 255:red, 74; green, 144; blue, 226 }  ,opacity=1 ]  {$\mathcal{D}$};
    % Text Node
    \draw (316.95,114) node [anchor=north west][inner sep=0.75pt]  [color={rgb, 255:red, 208; green, 2; blue, 27 }  ,opacity=1 ]  {$\mathcal{S}$};
    % Text Node
    \draw (404.48,130) node [anchor=north west][inner sep=0.75pt]    {$x$};
    % Text Node
    \draw (361.82,85.5) node [anchor=north west][inner sep=0.75pt]    {$\gamma ( x)$};
    % Text Node
    \draw (394.22,226.75) node [anchor=north west][inner sep=0.75pt]    {$\| x\| _{\mathcal{S}}$};
    % Text Node
    \draw (297.88,55) node [anchor=north west][inner sep=0.75pt]  [font=\scriptsize,color={rgb, 255:red, 65; green, 117; blue, 5 }  ,opacity=1 ]  {$\mathcal{D} \setminus \{x\in \mathcal{D} \mid \| x\| _{\mathcal{S}} < r\}$};
    % Text Node
    \draw (114.71,55) node [anchor=north west][inner sep=0.75pt]  [font=\scriptsize,color={rgb, 255:red, 99  ; green, 99; blue, 99}  ,opacity=1 ]  {$(\mathcal{D} \setminus \text{int}_{\mathcal{D}}\mathcal{S} )\setminus \{x\in \mathcal{D} \mid \| x\| _{\mathcal{S}}  >r\}$};
    % Text Node
    \draw (170,200) node [anchor=north west][inner sep=0.75pt]  [color={rgb, 255:red, 65; green, 117; blue, 5 }  ,opacity=1 ]  {$p( \| x\| _{\mathcal{S}})$};
    % Text Node
    \draw (355,249) node [anchor=north west][inner sep=0.75pt]    {$R_{0} =\overline{r}$};
    % Text Node
    \draw (242.35,248.75) node [anchor=north west][inner sep=0.75pt]    {$R_{1}$};
    % Text Node
    \draw (170.9,248.75) node [anchor=north west][inner sep=0.75pt]    {$R_{2}$};
    \draw (143,355.1) node [anchor=north west][inner sep=0.75pt]    {$r$};
    % Text Node
    \draw (396.7,330.75) node [anchor=north west][inner sep=0.75pt]    {$\| x\| _{\mathcal{S}}$};
    % Text Node
    \draw (385.62,189.75) node [anchor=north west][inner sep=0.75pt]    {$\phi _{0}$};
    % Text Node
    \draw (255.36,215) node [anchor=north west][inner sep=0.75pt]    {$\phi _{1}$};
    % Text Node
    \draw (315.65,214.5) node [anchor=north west][inner sep=0.75pt]    {$\alpha ( \| x\| _{\mathcal{S}})$};
    % Text Node
    \draw (339.75,294.5) node [anchor=north west][inner sep=0.75pt]  [color={rgb, 255:red, 99  ; green, 99; blue, 99 }  ,opacity=1 ]  {$q( \| x\| _{\mathcal{S}})$};
    % Text Node
    \draw (397.91,265.5) node [anchor=north west][inner sep=0.75pt]    {$\phi _{0}$};
    % Text Node
    \draw (353.82,352.25) node [anchor=north west][inner sep=0.75pt]    {$R_{0} =\overline{r}$};
    % Text Node
    \draw (241.82,352.5) node [anchor=north west][inner sep=0.75pt]    {$R_{1}$};
    % Text Node
    \draw (170,352.5) node [anchor=north west][inner sep=0.75pt]    {$R_{2}$};
    % Text Node
    \draw (255,266) node [anchor=north west][inner sep=0.75pt]    {$\phi _{1}$};
    % Text Node
    \draw (167,277) node [anchor=north west][inner sep=0.75pt]    {$\phi _{2}$};
    % Text Node
    \draw (188,268) node [anchor=north west][inner sep=0.75pt]    {$\beta ( \| x\| _{\mathcal{S}})$};
    
    %\draw   (510.19, 144) circle [x radius= 5, y radius= 5]   ;
     
    \end{tikzpicture}
    \caption{Illustrative example of the class $\mathcal{K}$ lower and upper bounds in Lemmata~\ref{lemma:clK_bound}~and~\ref{lemma:clK_upbound}. The function $\gamma(x)$ is minimised on $\mathcal{D}\setminus\{x\in\mathcal{D}\mid \Vert x \Vert_\mathcal{S} < r\}$ (green) and maximised on $(\mathcal{D}\setminus\textup{int}_{\mathcal{D}}\mathcal{S})\setminus\{x\in\mathcal{D}\mid \Vert x \Vert_\mathcal{S} > r\}$ (grey) to obtain the non-decreasing functions $p(\Vert x\Vert_{\mathcal{S}})$ and $q(\Vert x \Vert_{\mathcal{S}})$, respectively. Finally, the class $\mathcal{K}$ bounds are constructed via piecewise affine lower and upper bounds.}
    \label{fig:diagram}
\end{figure}
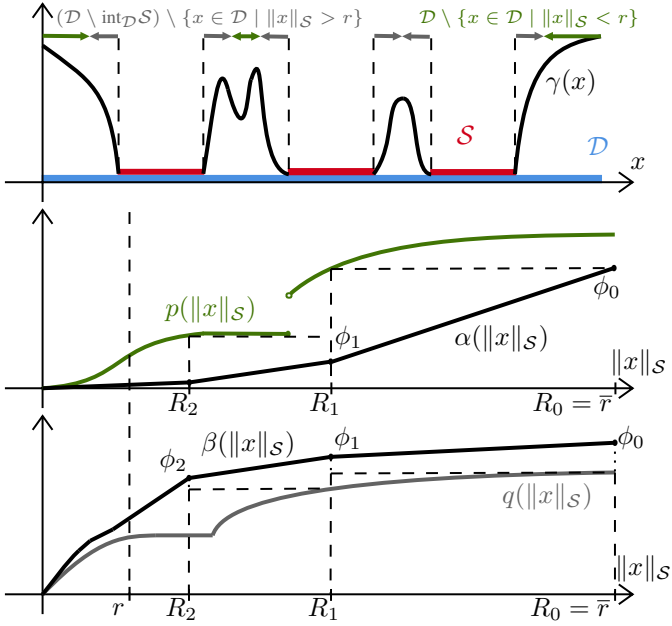

An illustrative example of the constructed upper and lower bounds is provided in Figure~\ref{fig:diagram}. 
We can now use the results in Lemmata~\ref{lemma:clK_bound}~and~\ref{lemma:clK_upbound}, to establish a relation between CBFs in Definitions~\ref{def:CBFcont}~and~\ref{def:CBFclK}.

\begin{theorem} \label{thrm:CBFs}
    Consider a CBF $h:\mathcal{D}\rightarrow \mathbb{R}$ according to Definition~\ref{def:CBFcont} and a  corresponding decrease function $\Delta h$. It holds that the function $\tilde{h}(x):=\max\{0,h(x)\}$, with $\tilde{h}:\mathcal{D}\rightarrow \mathbb{R}$, is a CBF according to Definition~\ref{def:CBFclK}.
\end{theorem}

\begin{proof}
    We prove the three statements in Definition~\ref{def:CBFclK} for $\tilde{h}(x)=\max\{0,h(x)\}$ based on the fact that $h$ satisfies Definition~\ref{def:CBFcont}: \\
     1) From Definition~\ref{def:CBFcont}, it follows that the safe set $\mathcal{S}=\{x\in\mathbb{R}^n\mid h(x) \leq  0\} = \tilde{\mathcal{S}}=\{x\in\mathbb{R}^n\mid \tilde{h}(x) \leq 0\}$ and $\mathcal{D}$ are compact and non-empty. \\
     2) From Lemma~\ref{lemma:clK_bound}, it follows that $\exists\; \alpha_1\in\mathcal{K}$, such that $\alpha_1(\Vert x \Vert_\mathcal{S})\leq h(x)=\tilde{h}(x)\; \forall x\in\mathcal{D}\setminus\mathcal{S}$. As $\tilde{h}(x)=0 \; \forall x \in \mathcal{S}$, it holds that $\alpha_1(\Vert x \Vert_\mathcal{S})\leq \tilde{h}(x)\; \forall x\in\mathcal{D}$. 
     From Lemma~\ref{lemma:clK_upbound}, as it holds that as $h(x)$ is continuous on $\mathcal{D}\setminus\interior\mathcal{S}$ and $h(x)=0\;\forall x\in\partial\mathcal{S}$ by definition, $\exists\;\alpha_2\in\mathcal{K}$ such that $h(x)\leq \alpha_2(\Vert x \Vert_\mathcal{S})\;\forall x\in\mathcal{D}\setminus\interior\mathcal{S}$. With $\tilde{h}(x)=0 \;\forall x\in\mathcal{S}$, it holds that $\tilde{h}(x)\leq\alpha_2(\Vert x \Vert_{\mathcal{S}}) \;\forall x\in \mathcal{D}$. \\
     3) We prove a class $\mathcal{K}$ decrease for $\tilde{h}$ for all $x\in\mathcal{D}$ by considering the following case distinction, denoting an optimal input as $u^*\in\arg\min_{u\in\mathcal{U}} \{h(f(x,u,0))\mid f(x,u,0)\in\mathcal{D}\}$:
     \begin{itemize}
        \itemindent=-8pt
        \item $\forall x\in \mathcal{D}\setminus\mathcal{S}$, we consider the two following cases
        \begin{itemize}      
        \itemindent=-16pt
        \item$f(x,u^*,0)\in\mathcal{D}\setminus\mathcal{S}:$\\ \hspace*{-0.55cm}$ \inf_{u\in\mathcal{U}}\{h(f(x,u,0)){\mid} f(x,u,0){\in}\mathcal{D}\}{-}h(x){\leq}{-}\Delta h(x)$ \\
        \hspace*{-0.55cm}${\Rightarrow}\! \inf_{u\in\mathcal{U}}\{\tilde{h}(f(x,u,0)) {\mid} f(x,u,0){\in}\mathcal{D}\}{\text{-}}\tilde{h}(x){\leq} {-}\Delta h(x)$. 
        \item[] As $\Delta h$ fulfils the conditions of Lemma~\ref{lemma:clK_bound}, it holds that\\ \hspace*{-0.68cm} $\exists \beta \in \mathcal{K}$, such that $\Delta h(x)\geq \beta(\Vert x \Vert_\mathcal{S}) \;\forall x\in\mathcal{D}\setminus\mathcal{S}$ and \\ \hspace*{-0.55cm}$\inf_{u\in\mathcal{U}}\{\tilde{h}(f(x,u,0)){\mid} f(x,u,0){\in}\mathcal{D}\}\text{-}\tilde{h}(x){\leq} {-}\beta(\Vert x\Vert_\mathcal{S})$.
            \item$f(x,u^*,0)\in\mathcal{S}:$ \\
            \hspace*{-0.55cm}It holds that $\inf_{u\in\mathcal{U}}\{\tilde{h}(f(x,u,0)) \mid f(x,u,0)\in\mathcal{D}\}{=}0$. \\ \hspace*{-0.55cm}Using the lower bound $\alpha_1(\Vert x \Vert_\mathcal{S})\leq \tilde{h}(x)$ from step 2), \\ \hspace*{-0.55cm}it holds that  \\ \hspace*{-0.55cm}$\inf_{u\in\mathcal{U}}\{\tilde{h}(f(x,u,0)) {\mid} f(x,u,0){\in}\mathcal{D}\}{-}\tilde{h}(x){\leq} {-}\alpha_1(\Vert x \Vert_\mathcal{S})$
        \end{itemize}
        \hspace*{-0.5cm} We can therefore define the function $\alpha_3\in\mathcal{K}$ with \\
        \hspace*{-0.5cm} $\alpha_3(\Vert x \Vert_\mathcal{S})=\min\{ \beta(\Vert x \Vert_\mathcal{S}),\alpha_1(\Vert x \Vert_\mathcal{S})\}$, such that
        $$\inf_{u\in\mathcal{U}}\{\tilde{h}(f(x,u,0)) \mid f(x,u,0)\in\mathcal{D}\}-\tilde{h}(x)\leq -\alpha_3(\Vert x \Vert_\mathcal{S})$$
        \item $\forall x\in\mathcal{S}:\; \inf_{u\in\mathcal{U}} h(f(x,u,0)) \leq 0$ \\ \hspace*{-0.5cm}$\Rightarrow \inf_{u\in\mathcal{U}}\{\tilde{h}(f(x,u,0)){\mid} f(x,u,0){\in}\mathcal{D}\}{-}\underbrace{\tilde{h}(x)}_{=0}{\leq}\underbrace{{-}\alpha_3(\Vert x \Vert_\mathcal{S})}_{=0}$ \\ \hspace*{-0.4cm} where we used $h(x)\leq\tilde{h}(x)= 0 \;\forall x\in\mathcal{S}$.
     \end{itemize}
     We therefore have $\forall x\in\mathcal{D}$, that 
     $$\inf_{u\in\mathcal{U}}\{\tilde{h}(f(x,u,0)) \mid f(x,u,0){\in}\mathcal{D}\}-\tilde{h}(x)\leq -\alpha_3(\Vert x \Vert_\mathcal{S}).$$
\end{proof}

Given that we have shown the relation between the two considered definitions for CBFs, it follows that the optimal value function $\hpb$ in \eqref{eq:pcbf} is also a CBF according to Definition~\ref{def:CBFclK}.
\begin{corollary}\label{cor:pcbf}
    Let Assumptions~\ref{ass:cont}~and~\ref{ass:terminalCBF} hold. The minimum \eqref{eq:pcbf} exists and if $\alpha_f$ is chosen sufficiently large, then the optimal value function $\hpb(x(k))$ defined in \eqref{eq:cbf_pcbf} is a (predictive) control barrier function according to Definition~\ref{def:CBFclK} with domain $\Dpb := \{ x \in \mathbb{R}^n \mid \hpb(x) \leq \alpha_f \gamma_f \}$ and safe set $\Spb := \{ x \in \mathbb{R}^n \mid \hpb(x) = 0  \}$.
\end{corollary}

\begin{proof}
    This result follows directly from Theorem~\ref{thrm:PCBF}, which shows that $h_{\mathrm{PB}}$ is a CBF according to Definition~\ref{def:CBFcont} under the given conditions. As $h_{\mathrm{PB}}=0$ for all $x\in\Spb$ and $h_{\mathrm{PB}}>0$ for all $x\in\Dpb\setminus\Spb$, it holds that $\hpb(x) = \max\{0, \hpb(x) \}$ and $\hpb$ is a CBF according to Definition~\ref{def:CBFclK} by Theorem~\ref{thrm:CBFs}.     
\end{proof}

\subsection{ISS of the Approximate PCBF Algorithm}
We can now establish theoretical guarantees of the closed-loop system under the proposed algorithm despite possible approximation errors and exogenous disturbances. First, we show that the optimisation problem~\eqref{eq:apcbf_sf} is feasible using the maximum decrease~\eqref{eq:apcbf_maxdec} for any $x\in\mathbb{R}^n$. Thereby a control input can always be computed using the provided algorithm, even if the state lies outside the theoretical region of attraction $\Dpb$, which may be a subset of the true region of attraction of the PCBF algorithm in~\cite{wabersich2022}.
\begin{proposition}
    The optimisation problem~\eqref{eq:apcbf_sf}, where $\Delta\ahpb(x)$ is defined in~\eqref{eq:apcbf_maxdec} is feasible for any $x(k)\in\mathbb{R}^n$ and $c_\Delta\in(0,1]$ if $\Delta \ahpb(x)\geq 0$ any $c_\Delta=1$ otherwise.
\end{proposition}

\begin{proof}
    The optimisation problem~\eqref{eq:apcbf_maxdec} is feasible for any $x\in\mathbb{R}^n$ by choosing any $u\in\mathcal{U}$. Given an optimal solution $u^*$ to~\eqref{eq:apcbf_maxdec}, it holds that $u^*$ is feasible in~\eqref{eq:apcbf_sf} as $u^*\in\mathcal{U}$ and for any $x\in\mathbb{R}^n$, $\ahpb(f(x,u^*,0)) - \ahpb(x) = - \max_{u\in\mathcal{U}} (\ahpb(x) - \ahpb(f(x,u,0))) \leq -c_\Delta \Delta \ahpb(x)$ for all $c_\Delta\in(0,1]$ when $\Delta\ahpb(x)\geq 0$ and for $c_\Delta=1$ if $\Delta\ahpb(x)< 0$.
\end{proof}
Next, we can show that the closed-loop system is ISS with respect to the approximation errors of $\ahpb$ in $\Dpb$ if the approximation error is bounded.
\begin{assumption}\label{ass:error}
    The approximation error of $\ahpb$ is uniformly bounded, i.e., it holds that $|\hpb(x)-\ahpb(x)| \leq \epsilon_h$ for all $x \in \Dpb$. In other words, it holds that $\ahpb(x) =\hpb(x) + e_{h}(x) $, where the error $e_{h}(x)$ is bounded by $|e_h(x)|\leq \epsilon_h$.
\end{assumption}

This assumption holds for any continuous regressor $\ahpb$, such as neural networks with continuous activation functions. Due to continuity of $\ahpb$ and $\hpb$ as well as compactness of $\Dpb$, a maximum $\epsilon_h$ exists for $\max_{x\in\Dpb} |\hpb(x)-\ahpb(x)|$ according to the extreme value theorem~\cite[Proposition A.7]{rawlings2017}. In order to analyse robust stability properties of the closed-loop, we require the following definition of robust positive invariance.
\begin{definition}[Definition 2.2 in \cite{blanchini1999}]
    A set $\mathcal{A}\subset\mathbb{R}^n$ is robustly positively invariant for a system $x(k+1)\in F(x(k),w(k))$ if it holds for all $w(k)\in\mathcal{W}$, for some compact set $\mathcal{W}$ containing all possible disturbance realisations, that
    \begin{equation*}
        x(k)\in\mathcal{A}\Rightarrow F(x(k),w(k))\subseteq \mathcal{A}.
    \end{equation*}
\end{definition}

For ISS, we use the following comparison function definition.

\begin{definition}[Definition B.42 in \cite{rawlings2017}]\label{def:ISS}
    The system $x(k+1)\in F(x(k),w(k))$ is input-to-state stable in $\mathcal{D}$ with respect to the compact, non-empty and $0$-input invariant set\footnote{A set $\mathcal{S}\subset \mathcal{D}$ is denoted $0$-input invariant for the dynamics $x(k+1)\in F(x(k),w(k))$ if it is invariant for the dynamics $x(k+1)\in F(x(k),0)$.} $\mathcal{S}$ if there exists a $\mathcal{KL}$ function\footnote{According to \cite[Definition B.3]{rawlings2017}, a function $\beta:\mathbb{R}_{\geq0}\times \mathbb{N}_{\geq0}\rightarrow\mathbb{R}_{\geq0}$ belongs to class $\mathcal{KL}$ if it is continuous and if, for each $t\geq 0$, $\beta(\cdot, t)$ is a class $\mathcal{K}$ function and for each $s\geq0$, $\beta(s,\cdot)$ is non-increasing and satisfies $\lim_{t\rightarrow\infty}\beta(s,t)=0$.} $\beta$ and a $\mathcal{K}$ function $\sigma$ such that, for each $x\in\mathcal{D}$, with $\mathcal{D}$ compact and robustly positively invariant, and each disturbance sequence $\mathbf{w}=(w(0), w(1), \dots)$ in $\ell_\infty$
    \begin{equation*}
        \Vert\phi(i;x(0),\mathbf{w}_i)\Vert_\mathcal{S} \leq \beta(\Vert x\Vert_\mathcal{S},i) + \sigma(\Vert \mathbf{w}_i \Vert)
    \end{equation*}
    for all $i\in\mathbb{N}_{\geq 0}$, all solutions of the system $\phi(i;x(0),\mathbf{w}_i)$ at time $i$ if the initial state is $x(0)$, and the input sequence is $\mathbf{w}_i:=(w(0), w(1), \dots, w(i-1))$.
\end{definition}
We note that with regards to input-to-state stability, the disturbances $w(k)$ are considered as the input of the system. Following the definition of ISS, a corresponding ISS Lyapunov function can be defined, which we denote as ISS control barrier function in this paper to distinguish stability of singletons from stability of sets.
We use the following definition to characterize an ISS-CBF, implying ISS of the closed-loop as shown in \cite{jiang2001}. 
\begin{definition}[Definition B.43 in~\cite{rawlings2017}] \label{def:ISSCBF}
    Let $\mathcal{S}$ and $\mathcal{D}$ be compact and non-empty, with $\mathcal{S}$ $0$-input invariant, $\mathcal{D}$ robustly positively invariant and $\mathcal{S}\subset\mathcal{D}$. A function $h:\mathcal{D}\rightarrow \mathbb{R}$ is an ISS-CBF for system $x(k+1)\in F(x(k),w(k))$ if there exist class $\mathcal{K}$ functions 
    $\alpha_1$, $\alpha_2$, $\alpha_3$ and $\sigma$ such that for all $x\in\mathcal{D}, w\in\mathbb{R}^p$, 
    \begin{subequations}
    \begin{align}
        \alpha_{1}(\Vert x \Vert_\mathcal{S}) &\leq h(x) \leq \alpha_2(\Vert x \Vert_\mathcal{S}) \label{eq:ISSbounds} \\
        \sup_{x^+\in F(x,w)} h(x^+) &- h(x) \leq -\alpha_3(\Vert x \Vert_\mathcal{S}) + \sigma(\Vert w \Vert) \label{eq:ISSdec}
    \end{align}
    \end{subequations}
\end{definition}
The following Lemma establishes the relationship between the existence of an ISS-CBF and ISS of the system.
\begin{lemma}[\cite{jiang2001}] \label{lemma:ISSequiv}
    Suppose $\tilde{f}$ is continuous and that there exists a continuous ISS-CBF according to Definition~\ref{def:ISSCBF} for $x(k+1)\in F(x(k),w(k))$. Then the system $x(k+1)\in F(x(k),w(k))$ is ISS.
\end{lemma}

\begin{remark}
    Note that in \cite[Appendix B]{rawlings2017}, the definitions of ISS and of an ISS-CBF are given with respect to the singleton $\mathcal{S}=\{0\}$, however this result is applicable also for compact sets $\mathcal{S}$, see, e.g., \cite[Remark 3.12]{jiang2001}. Additionally, the results are given for the difference equation $x(k+1)=f(x(k),w(k))$ rather than the difference inclusion $x(k+1)\in F(x(k),w(k))$, but this difference is minor due to the supremum in the decrease condition \eqref{eq:ISSdec}.
\end{remark}

By showing the existence of an ISS-CBF implies that the closed-loop system under the proposed approximate PCBF algorithm is ISS according to~\cite[Definition B.42]{rawlings2017}.
\begin{theorem}\label{thm:max_decrease}
    Let Assumptions~\ref{ass:cont},~\ref{ass:terminalCBF}~and~\ref{ass:error} hold. If $\Dpb$ is robustly positively invariant for the approximation $\hat{h}_{\mathrm{PB}}$ and disturbances $w(k)$, system~\eqref{eq:sys} under application of Algorithm~\ref{alg:apcbf} is ISS.
\end{theorem}

\begin{proof}
    We show that $\hpb(x)$ is an ISS-CBF according to Definition~\ref{def:ISSCBF} for the closed-loop system $x(k+1)\in \hat{F}(x(k),w(k)) = \{x^+\in\mathbb{R}^n \mid x^+ = f(x(k),u,w(k)), u\in\mathcal{U}_{\ahpb}(x(k))\}$, where we define $\mathcal{U}_{\ahpb}(x(k))\coloneqq\{u\in\mathcal{U}\mid \ahpb(f(x(k),u, 0)) \!-\! \ahpb(x(k)) \leq - c_\Delta \Delta\ahpb(x(k)) \}$. Given that inputs are selected as an optimal solution to~\eqref{eq:apcbf_sf}, it holds that under application of Algorithm~\ref{alg:apcbf}, $u(k)\in\mathcal{U}_{\ahpb}(x(k))$. As $\hpb$ is a CBF according to Definition~\ref{def:CBFclK}, the condition~\eqref{eq:ISSbounds} is satisfied. Furthermore, we can show that for any $u\in\mathcal{U}_{\ahpb}(x)$, it holds $\forall w\in\mathbb{R}^p$ that 
    \begin{align*}
        &\sup_{x^+\in\hat{F}(x,w)}\hpb(x^+) - \hpb(x) \\
        \leq & \ahpb(f(x,u,w))+|e_h(f(x,u,w))|-\ahpb(x)+|e_h(x)|  \\
        \leq & \ahpb(f(x,u,0)) - \ahpb(x) + \sigma_h(\Vert w \Vert) + |e_h(f(x,u,0))| \\ & +\sigma_e(\Vert w \Vert) + |e_h(x)| \\
        \leq & -c_\Delta\Delta\ahpb(x) + 2 \epsilon_h + \sigma(\Vert w\Vert) \\
        \leq & -c_\Delta \max_{u\in\mathcal{U}} \left( \ahpb(x)-\ahpb(f(x,u,0))\right) + 2 \epsilon_h + \sigma(\Vert w\Vert) \\
        \leq & -c_\Delta \Delta h(x) + 4 \epsilon_h + \sigma(\Vert w \Vert),
    \end{align*}
    where the class $\mathcal{K}$ functions $\sigma_h$ and $\sigma_w$ exist due to uniform continuity of $\hat{h}_{\mathrm{PB}}$ and $f$ on the compact domain $\Dpb$ as shown in~\cite{limon2009}. The set $\Spb$ is $0$-invariant with $\epsilon_h=0$ and $w=0$ and the set $\Dpb$ is robust invariant, if, e.g., $\epsilon_h$ and $w$ is sufficiently small, i.e., it holds that $4\epsilon_h+\sigma(\Vert w \Vert)\leq \alpha_f\gamma_f - \bar{c}$, where $\bar{c}\coloneqq \max_{x\in\Dpb} \hpb(x)- c_\Delta \alpha_3\circ\alpha_1^{-1}(\hpb(x))$, which ensures that $\hpb(x^+)\leq\alpha_f\gamma_f$ for all $x\in\Dpb$. 
\end{proof}

\begin{remark}\label{rem:approx_errors}
    In the proof of Theorem~\ref{thm:max_decrease}, the error $|e_h(x)|$ can also be bounded by a class $\mathcal{K}$ function $\gamma$ and a constant term $\tilde{\epsilon}_h\leq \epsilon_h$, such that $|e_h(x)|\leq \gamma(x) + \tilde{\epsilon}_h$. If it holds that $\forall x\in\Dpb$, $c_\Delta \Delta \hpb(x) > 2\gamma(x)$, the decrease condition is still satisfied with respect to the smaller constant $2\tilde{\epsilon}_h+2\epsilon_h$. This implies that approximation errors close to $\partial \Spb$ have a bigger impact on the closed-loop stability than errors when $\hpb$ is large, which can be leveraged when performing the approximation as discussed in Section~\ref{sec:NE}.
\end{remark} 

This analysis applies to any continuous approximation of a CBF according to Definition~\ref{def:CBFcont}~or~\ref{def:CBFclK}. It is especially useful when formal verification methods as discussed in~\cite{liu2021} are not applicable due to large networks and when exogenous disturbances are present. We note that for many neural network-based CBF approximations, guaranteeing that a decrease condition is satisfied is encouraged by specific loss functions during training, as discussed in~\cite{so2024}, but cannot typically be directly enforced. If the optimal value function $\hpb$ is Lipschitz continuous, the approach in~\cite{robey2020} could be applied to obtain a valid approximation, however the resulting optimization problem suffers from computational complexity issues. 

The established ISS result therefore provides an understanding of the closed-loop system behaviour when the CBF decrease constraint is possibly violated, either through approximation errors or disturbances. It implies that the system converges to and remains in a neighbourhood of the safe set. Additionally, sampling-based approaches as used in~\cite{hertneck2018} can provide high probability guarantees for invariance of $\Dpb$ given $\ahpb$ instead of the provided error bound. 

The proposed result directly implies ISS of the PCBF algorithm in~\cite{wabersich2022} as it holds that $\epsilon_h=0$. Furthermore, if $\Delta\ahpb(x)>0$ for all $x\in \hat{\mathcal{D}} \setminus \hat{\mathcal{S}}$, with $\hat{\mathcal{S}}=\{x\mid \ahpb(x)\leq 0\}$ and $\hat{\mathcal{D}}=\{x \mid \ahpb(x)\leq \hat{\gamma}\}$ for $\hat{\gamma}>0$, then $\max\{0, \ahpb(x)\}$ directly defines a CBF according to Definition~\ref{def:CBFclK}.

\section{NUMERICAL EXAMPLES} \label{sec:NE}
\subsection{Unstable Linear System}
As an illustrative example, we consider an unstable linearised pendulum with a damper component
\begin{equation*}\label{eq:syslin}
x(k+1)  =\begin{bmatrix}
    1 & T_s \\ T_s\frac{g}{l} & 1-T_sc
\end{bmatrix} x(k) + \begin{bmatrix}
    0 \\ 5T_s
\end{bmatrix} u(k),  \nonumber
\end{equation*}
where the states are given by the angle and angular velocity, i.e., $x=\left[ \psi\;\; \omega  \right]^\top \in\mathbb{R}^2$, the input $u\in\mathbb{R}$ and the dynamics are defined through the sampling time $T_s=0.1$, the damping coefficient is $c=1$, the length $l=1.3$ and $g=9.81$. The system is subject to polytopic state and input constraints $x(k)\in\mathcal{X}=\{x\in\mathbb{R}^2 \mid A_x x \leq b_x\}$ and $u(k)\in\mathcal{U}=\{u\in\mathbb{R} \mid A_u u \leq b_u\}$, which are given as bounds on the angle $-0.5 \leq \psi \leq 0.5$, the angular velocity $-1.5\leq \omega \leq 1.5$ and the input $-3\leq u \leq 3$. We choose the terminal penalty $\alpha_f=1$e6 and tightening values $\Delta_i=i\cdot 1$e-3 for $i=1,\dots, N$, with prediction horizon $N=10$.

\begin{table}[!t]
    \centering
    \vspace{0.0cm}
    \caption{Normalised magnitude of safety interventions for safety filters based on the pcbf and two underfit neural network approximations for 10000 random states}
    \begin{tabular}{|c|c|c|c|c|c|}
        \hline
        & PCBF~\cite{wabersich2022} &\begin{tabular}{@{}c@{}}  $\ahpb$ \\ $c_\Delta{=}1$\end{tabular}  &  
        \begin{tabular}{@{}c@{}}$\hat{h}_{\mathrm{CE}}$ \\ $c_\Delta{=}1$\end{tabular} & \begin{tabular}{@{}c@{}}$\ahpb$ \\ $c_\Delta{=}0.2$\end{tabular} & \begin{tabular}{@{}c@{}}$\hat{h}_{\mathrm{CE}}$ \\ $c_\Delta{=}0.2$\end{tabular} \\ \hline
        $\Dpb \setminus\Spb$ & 1 & 1.05 & 1.10 & 0.41 & 0.48 \\ \hline
        $\Spb$ & 1 & 9.33 & 3.50 & 1.92 & 1.35 \\ \hline
    \end{tabular}
    \label{tab:sf_int}
    \vspace{-0.5cm}
\end{table}

In order to compute the terminal CBF $h_f$, we follow a similar procedure to \cite[Section IV]{wabersich2022} by first computing a quadratic Lyapunov function $V(x)=x^\top P x$ and corresponding state feedback $u(k)=Kx(k)$ for the system \eqref{eq:syslin}. The Lyapunov function is computed such that the ellipse described by $V(x)\leq 1$ is the largest ellipse within the tightened state constraints $\bar{X}_{N-1}(0)$ such that $Kx\in\mathcal{U}$ according to the following semi-definite program
\begin{subequations}
    \begin{align}
    \min_{E,Y} & -\log\det(E) \\
    \textup{s.t. } & E\succeq 0 \\
    &\begin{bmatrix} E & (AE+BY)^\top \\ AE+BY & E
    \end{bmatrix} \succeq 0 \\
    & \begin{bmatrix} ([b_x]_i - \Delta_{N-1})^2 & [A_x]_iE \\ ([A_x]_iE)^\top & E 
    \end{bmatrix} \succeq 0 \;\; \forall i=1,\dots 4 \\
    & \begin{bmatrix} [b_u]_j^2 & [A_u]_jY \\ ([A_u]_jY)^\top & E 
    \end{bmatrix} \succeq 0 \;\; \forall j=1,\dots 2,
    \end{align}
\end{subequations}
similarly to, e.g., \cite[Section IV]{wabersich2022},
such that $P$ and $K$ can by computed as $P=E^{-1}$ and $K=YP$. The terminal CBF is then given by $h_f=V(x)-1$, such that $\mathcal{S}_f$ fulfils Assumption~\ref{ass:terminalCBF} and the domain $\mathcal{D}_f$ is given as the largest sublevel set $h_f(x)\leq\gamma_f\approx 2.95$ for which the input constraints $u=Kx\in\mathcal{U}$ are satisfied for all $x\in\mathcal{D}_f$.

The training and validation data sets $\mathbb{D}=\{(x_i, h(x_i))\}$ and $\mathbb{D}_{\textup{PB}}=\{(x_i, \hpb(x_i))\}$ are collected by gridding the state space in $-2\leq \psi \leq 2$ and $-8\leq \omega \leq 8$ as well as additional data points close to the constraints in $\{x\mid 0.8 b_x\leq A_x x\leq 1.2b_x\}$ using a rejection sampler for $\hpb(x_i)\leq 10$, resulting in approximately $3.5\textup{e}5$ data points. 
The data collection for the training data of the extended PCBF took approximately $45$min without parallelization on an Intel Core i9-7940X CPU @ $3.10$GHz.
In order to illustrate the advantages of the continuous extension of the PCBF, we approximate $\hpb$ and $h$ such that the models $\hat{h}_{\mathrm{CE}}$ and $\ahpb$ underfit the data. In order to achieve this underfit, we train two neural networks using PyTorch \cite{paszke2019} on $\mathbb{D}$ and $\mathbb{D}_{\textup{PB}}$ with two hidden layers of 4 neurons using the softplus activation function. The networks are trained for $50$ epochs using a learning rate of $1\textup{e}-3$, a batch size of $128$ and the mean squared error (MSE) as a loss function. Note that the expressiveness of the chosen networks is not high enough to accurately fit the data, such that after training the MSE with respect to the validation set is given by $0.25$ and $0.11$ for $\hat{h}_{\mathrm{CE}}$ and $\ahpb$, respectively. In order to compare the performance of both approximations in the safety filter formulation, we sample $1\textup{e}4$ random states within $\mathcal{X}$. The total safety filter interventions in terms of the norm of the difference between the computed input and a performance input $u_\mathrm{p}(k)=0$, normalised by the interventions of the PCBF algorithm in \cite{wabersich2022}, is given in Table~\ref{tab:sf_int}. 
We observe that the underfit models $\ahpb$ and $\hat{h}_{\mathrm{CE}}$ using $c_\Delta=1$ show a similar magnitude of safety interventions in $\Dpb\setminus\Spb$, being more conservative with $5\%$ and $10\%$ more interventions, respectively, compared to the full optimal control based algorithm. However, within $\Spb$, the approximation $\ahpb$ shows a $9$ times increase in safety interventions due to learning errors when $\hpb(x)=0$ possibly leading to positive values of $\ahpb$, which results in modifying the proposed input. The continuous extension reduces the impact of these errors by approximating negative values in $\Spb$, such that the safety interventions are only $3.5$ times as large. As the explicit approximation enables the introduction of the additional hyperparameter $c_\Delta$ to tune the rate of PCBF decrease in $\Dpb\setminus\Spb$, we also compare the magnitude using $c_\Delta=0.2$. It can be seen that this results in a significant decrease in interventions in $\Dpb\setminus\Spb$, reducing the magnitude by a factor of $2$ for both approximations.

Finally, we compare the closed-loop performance of both algorithms in Figure~\ref{fig:linear}. For this comparison, we train another neural network using the same data sets generated by sampling $h(x)$ by gridding the state space, where we choose a higher expressiveness of the model by using $3$ hidden layers with $16$ softplus activation functions. The model is trained in $30$ epochs using a batch size of $8$, a learning rate of $5\textup{e}-3$ and the MSE loss function, with a total computation time of approximately $25$min. 
The final validation MSE is given by $6\textup{e}-3$, which is two orders of magnitude lower compared to the two underfit models. Following the arguments in \cite{hertneck2018}, with a confidence of 1-$1$e-5, the probability that the domain $\{x\mid \hpb(x)\leq 10\}$, where training data was collected, is robustly invariant is greater than $99.75\%$. As a performance input, we choose $u_\mathrm{p}(k)=2$, which aims to drive the system into the constraint $\psi=0.5$ in an adversarial sense and causes a slight chattering in the velocity at the constraint in closed-loop. In practice, such chattering can be reduced using an additional regularisation penalising the change in input as done in~\cite{tearle2021}. By setting $c_\Delta=0.15$, we observe that in closed-loop, for 5 different initial conditions using the same initial angle $\psi=0$ and different velocities $\omega$, the system is allowed to converge much slower to $\Spb$, while still converging to the constraint $\psi=0.5$. 
We additionally provide a comparison to a soft-constrained predictive safety filter (PSF) in Figure~\ref{fig:linearPSF} for an initial velocity of $\omega=5$. The soft-constrained PSF exhibits the same behaviour as the PCBF algorithm when choosing a large slack penalty, denoted as $c_\xi$, in the cost and does not converge to the constraints for low values. The average solve times over the $150$ simulated time steps are given by $0.017$s for the soft-constrained PSF, $0.021$s for the PCBF algorithm in~\cite{wabersich2022} and $0.011$s for the approximate PCBF safety filter, which are solved using CasADi~\cite{andersson2019} and IPOPT~\cite{biegler2009}. We note that even for the small time horizon of $N=10$, we observe a computational speedup of over $50\%$. While the computation time of the PCBF algorithm and soft-constrained PSF are expected to increase with increasing time horizon, the online computation of the approximation based safety filter is independent of the considered time horizon and only the offline data collection time will be impacted.

\begin{figure}[!t]
\centering
\includegraphics[width=\columnwidth]{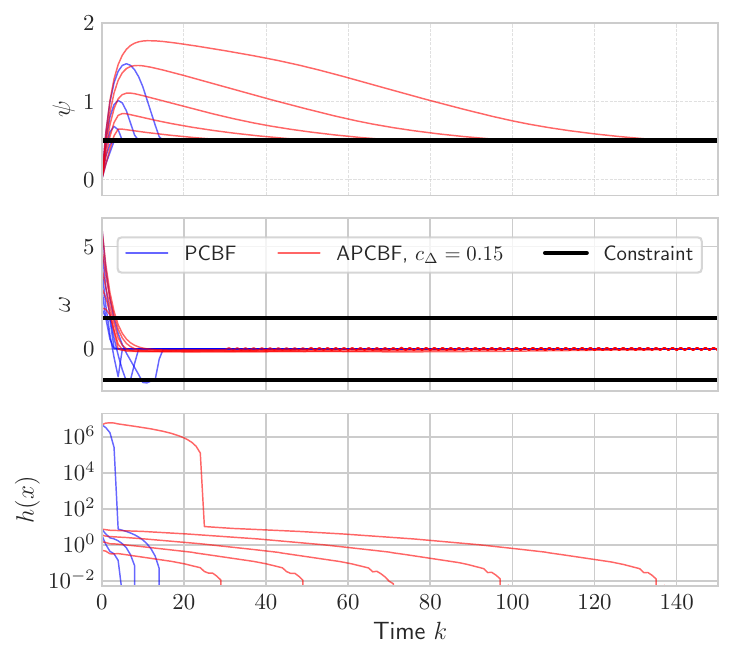}
    \caption{Closed-loop comparison of the PCBF-based safety filter in \eqref{eq:pcbf} and \eqref{eq:pcbf_sf} with the proposed approximation scheme for system \eqref{eq:syslin} with states given by the angle $\psi$ (top) and the angular velocity $\omega$ (middle). The explicit approximation enables tuning the rate of the PCBF decrease through the tuning parameter $c_\Delta\in(0,1]$, resulting in a slower convergence of the states to $\Spb$, seen by the corresponding CBF value $h(x)$ (bottom).}
    \label{fig:linear}
\end{figure}

\begin{figure}
\centering
\includegraphics[width=\columnwidth]{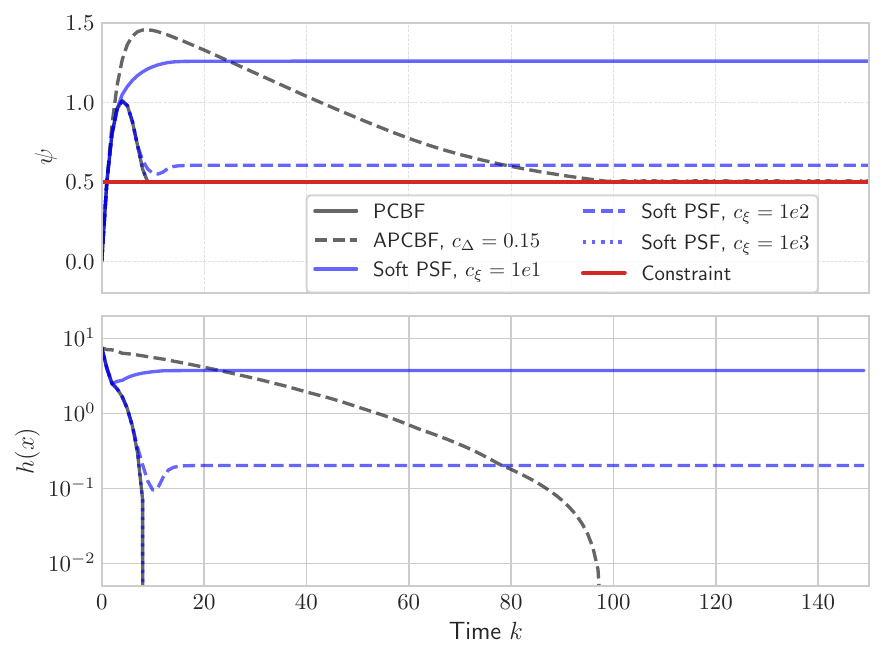}
    \vspace{-0.8cm}
    \caption{Closed-loop comparison of soft-constrained predictive safety filters and the PCBF-based safety filter in~\eqref{eq:pcbf} and~\eqref{eq:pcbf_sf} with the proposed scheme for an unstable pendulum. The angle $\psi$ (top) is depicted given a high initial angular velocity. The proposed formulation enables tuning the rate of the PCBF decrease through the tuning parameter $c_\Delta\in(0,1]$, resulting in a slower convergence of the states to $\Spb$ as can be observed from the PCBF values (bottom).}
    \vspace{-0.5cm}
    \label{fig:linearPSF}
\end{figure}

\subsection{Nonlinear System}

In this section, we demonstrate the potential for the proposed approximation-based algorithm as a computationally efficient safety filter. We consider a kinematic miniature race car model discretised using Euler forward and inspired by \cite{carron2023} of the form
\begin{align}\label{eq:syscar}
x(k+1)&=f(x(k),u(k)) \nonumber\\
&=\begin{bmatrix}
p_x(k) + T_s(v(k)\cos(\psi(k))) \\
p_y(k) + T_s(v(k)\sin(\psi(k))) \\
\psi(k) + T_s(v(k)/l\tan(\delta(k))) \\
v(k) + T_s u_0(k)\\
\delta(k) + T_s u_1(k)
\end{bmatrix},
\end{align}
where the length of the car is $l=0.09$m and the sampling time is $T_s=0.02$s. The states are given by the $x$- and $y$-position of the car, the heading angle $\psi$, the velocity $v$ and the steering angle $\delta$, i.e. $x=\left[p_x, p_y,\psi,v,\delta \right]^\top\in\mathbb{R}^5$ and the inputs are given by an acceleration and the steering change, resulting in $u\in\mathbb{R}^2$. The input constraints are given by $-5\leq u_0\leq 2$ and $-2.8\leq u_1\leq 2.8$ and the steering lies within $-0.75\leq \delta\leq 0.75$. The car operates in a cluttered environment with 9 obstacles, which are defined by $(p_x-o_{x,i})^2 + (p_y-o_{y,i})^2 \geq r_i^2$ for each obstacle $i$ with varying centre position $(o_{x,i},o_{y,i})$ and radius $r_i$. Finally, there are no constraints on the other states. For compactness of $\mathcal{D}$ and $\mathcal{S}$, they can be assumed to be large enough that they remain inactive during operation.

We consider an increasing tightening $\Delta_j=j\cdot1\textup{e}-4$ over a prediction horizon of $N=50$. As a terminal constraint, we impose the car to stop within the constraints $c_x(x)-\Delta_N\leq 0$ described above with a velocity $v=0$ and these terminal constraints are not softened with a terminal slack, i.e. $\xi_N=0$. 

The training and validation data sets $\mathbb{D}_{\textup{tr}}=\{(x_i,h(x_i)\}$ and $\mathbb{D}_{\textup{val}}$ are collected through the geometric sampler as described in Algorithm~\ref{alg:sampling}, with different initial seed and goal points such that the data is independent. The data is collected using $48$ cores on the ETH Euler cluster.
The number of collected training and validation points is approximately $1.2\textup{e}7$ and $6\textup{e}6$, with a computation time of $20$h and $16$h, respectively.
Given the fact that the approximation error for the negative values within $\Spb$ and for large values of $\hpb$ have a lower impact on the stability guarantees as also discussed in Remark~\ref{rem:approx_errors}, we preprocess the data by applying a cubic root to the function output, i.e. $y_i=\sqrt[3]{h(x_i)}$. Thereby, negative values in $\Spb$ remain negative and large positive values are decreased, while values close to $\partial\Spb$ are increased, such that approximation errors in this region incur a higher loss. We train a neural network with $5$ hidden layers of $32$ softplus activation functions on the preprocessed data using $20$ epochs, a learning rate of $2\textup{e}-2$ and a batch size of $128$ using a Huber loss function. The training time was $2.2$h and the final validation error is given by approximately $1\textup{e}-2$. Note that in principle the loss could be modified to prevent false positives in terms of safety certification of the inputs, as similarly done in \cite{lavanakul2024safety}, by encouraging the regressor to upper bound the true PCBF, but the approximation in this example was sufficiently accurate.

In order to demonstrate the computational advantages as a safety filter of the proposed approach, we use a naive model predictive control implementation as a performance controller, which is given as
\begin{align*}
    u_\mathrm{p}(k)\in\arg\min_{u_{i|k}} & \sum_{i=0}^{30} \Vert p_{x,i|k} +p_x(0)\Vert^2_2 + \Vert p_{y,i|k}+p_y(0)\Vert^2_2\\
    & \hspace{0.22cm}+ \Vert v_{i|k} - v_{\textup{ref}}\Vert^2_2 \\
    \textup{s.t. }& \forall i=0,\dots,29\\
    & x_{i+1|k}=f(x_{i|k},u{i|k}) \\
    &u_{i|k} \in\mathcal{U} \\
    &x_{0|k}=x(k) \\
\end{align*}
This implementation aims to the reflection of the initial position through the origin, while keeping the velocity constant at $v_{\textup{ref}}=3\textup{ms}^{-1}$. Notably, the optimisation problem is not aware of the obstacles, such that obstacle avoidance is purely a result of the safety filter. The closed-loop results for $5$ initial conditions over $70$ time steps are shown in Figure~\ref{fig:car}, where the system successfully avoids the obstacles in the cluttered environment, even at the high initial speed of $3\textup{ms}^{-1}$. The solve time for the neural network inference for safety verification, i.e. lines 4 and 5 in Algorithm~\ref{alg:apcbf}, take on average $0.26$ms and solving the two optimisation problems takes on average $37$ms. Comparatively, the PCBF algorithm in \cite{wabersich2022} requires on average $500$ms to solve the two predictive optimisation problems. The computational speedup is therefore a factor of $13$ in the case of overwriting the inputs and $1900$ when the proposed input is deemed safe. Finally, we provide the magnitude of the safety interventions for one of the trajectories in Figure~\ref{fig:inter}. It can be seen that the safety filter not only overrides the steer change before an obstacle is hit, but also the acceleration such that the car is able to navigate around the obstacles safely.

\begin{figure}[!t]\label{fig:car}
\vspace{0.1cm}
\centering
\includegraphics[width=\columnwidth]{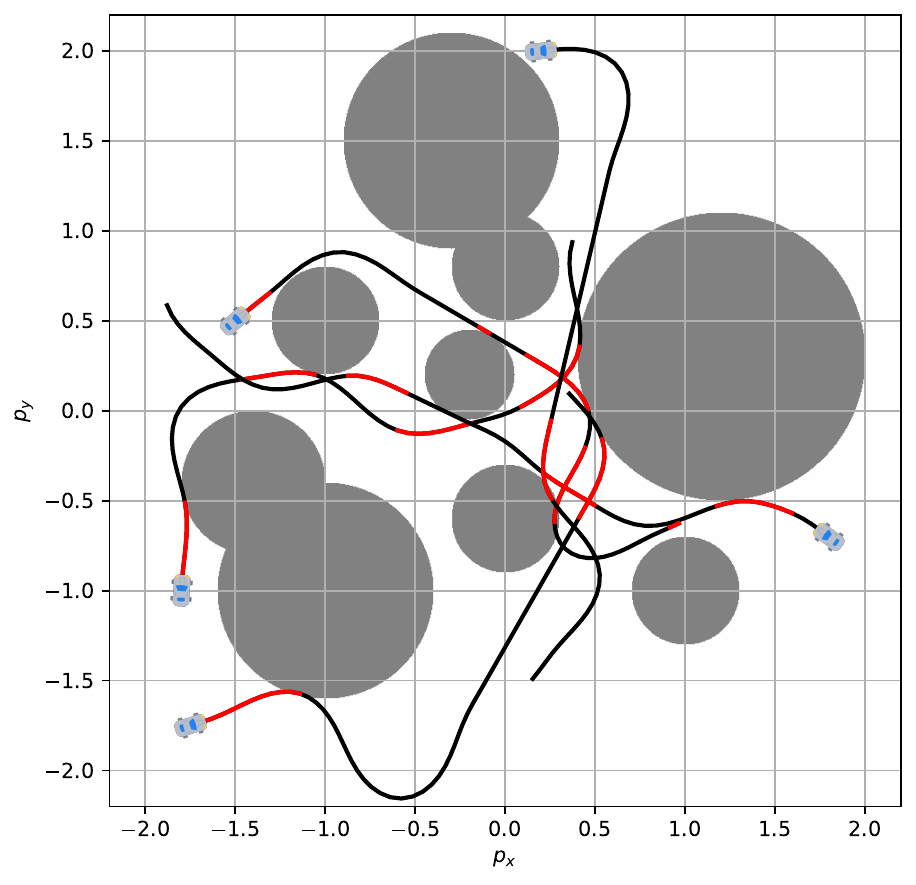}
    \caption{Closed-loop trajectories of the simulated car dynamics, where multiple initial conditions with a velocity of $3$m/s and an initial steering angle of $\delta=0$ are shown (car). The performance controller is given by a naive MPC implementation which is not aware of the obstacles, such that the distance to the opposite side of the area is minimised. States where the safety filter intervened are marked in red.}
\end{figure}

\begin{figure}[!t]\label{fig:inter}
\centering
\includegraphics[width=\columnwidth]{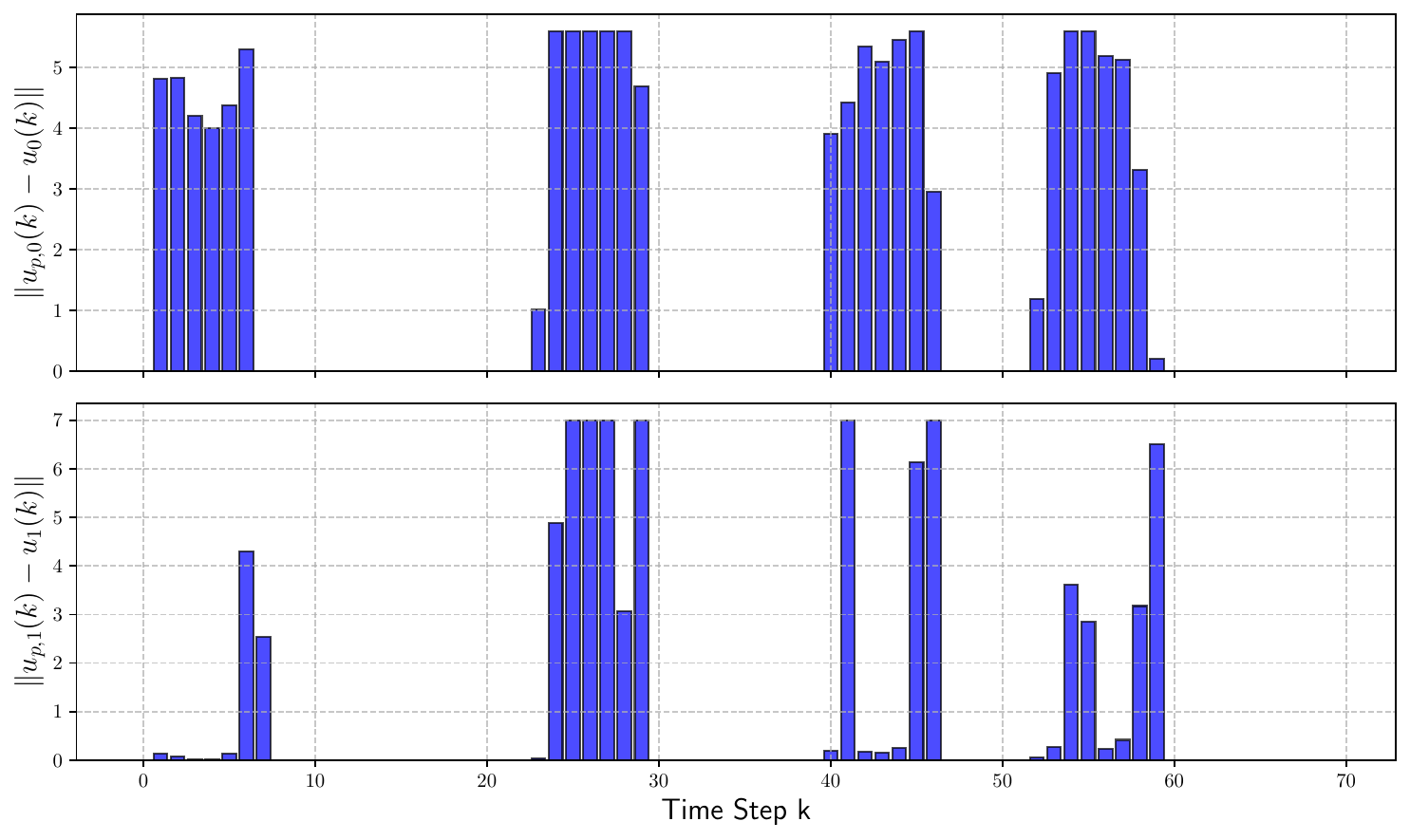}
    \caption{Magnitude of safety filter interventions of the steer change (top) and the acceleration (bottom) for the trajectory with initial position $(-1.8, -1.0)$.}
\end{figure}
\section{CONCLUSION} \label{sec:concl}
We extend initial results in~\cite{didier2023}, introducing an approximation of the PCBF in a CBF-based safety filter. 
We advance the theoretical analysis with respect to exogenous disturbances of the algorithms in~\cite{wabersich2022, didier2023} and propose a negative continuous extension of the predictive control barrier function to mitigate approximation errors within the safe set. The theoretical properties are validated in an illustrative linear example and the practicality of a computationally advantageous safety filter is shown in a miniature race car simulation. Future work could involve integrating the PCBF framework in a reinforcement learning setting similar to~\cite{fisac2019} to scale the training to possibly higher state dimensions as well as including contextual information of the environment.

\begin{appendix}\label{app}
In order to achieve a computationally efficient sampling of the extended PCBF $h_{\textup{CE}}$ in \eqref{eq:PCBFext}, we propose a modification of the geometric sampling method presented in \cite{chen2022} in Algorithm~\ref{alg:sampling}. The geometric sampling method allows using previously computed solutions to the predictive optimisation problems for the computation of $h_{\textup{CE}}$ as a warmstart as well as storing information if the previous sample on a line search was within $\Spb$ or $\Dpb\setminus\Spb$, reducing the number of times two optimisation problems have to be solved.
\begin{algorithm}[!t]
\caption{Sampling $h_{\textup{CE}}$}\label{alg:sampling}
\begin{algorithmic}[1]
\STATE Generate goal points $\mathbb{G}$ according to \cite[Algorithm 2]{chen2022}
\STATE Initialise boolean $\mathrm{in\_}\mathcal{S}_{\mathrm{PB}}\leftarrow\mathrm{False}$
\STATE Choose initial seed $\mathbb{S}\leftarrow\{\mathbf{s}\} = \{ (x, \mathbf{o},\mathrm{in\_}\mathcal{S}_{\mathrm{PB}}) \}$
\STATE Initialise data set $\mathbb{D}\leftarrow\emptyset$
\STATE Choose step size $d>0$
\FOR{goal point $g\in \mathbb{G}$ (in parallel)} 
    \STATE Sample seed $\mathbf{s}_0\leftarrow(x_0, \mathbf{o}_0, \mathrm{in\_}\mathcal{S}_{\mathrm{PB},0})\in\mathbb{S}$
    \STATE Compute number of points on line $n\leftarrow\Vert x_0-g\Vert/d$
    \FOR{$i=1, \dots, n$}
        \STATE $x_i\leftarrow x_0 + i (g-x_0)/n$
        \IF{$\mathrm{in\_}\mathcal{S}_{\mathrm{PB},i-1}$}
            \STATE Obtain $\mathbf{o}_i$, $h_{\textup{CE}}(x_i)$ from \eqref{eq:optprobwithnegativeslacks} with warmstart $\mathbf{o}_{i-1}$
            
            \IF{infeasible}
                \STATE $\mathrm{in\_}\mathcal{S}_{\mathrm{PB},i}\leftarrow$ False
                \STATE Obtain $\mathbf{o}_i$, $h_{\textup{CE}}(x_i)$ from \eqref{eq:pcbf} with warmstart $\mathbf{o}_{i-1}$
            \ELSE 
                \STATE $\mathrm{in\_}\mathcal{S}_{\mathrm{PB},i}\leftarrow$ True
            \ENDIF
            \STATE $\mathbf{s}_i \leftarrow (x_i, \mathbf{o}_i, \mathrm{in\_}\mathcal{S}_{\mathrm{PB},i})$
            \STATE $\mathbb{D}\leftarrow \mathbb{D}\cup (x_i, h_{\textup{CE}}(x_i))$
        \ELSE
            \STATE Obtain $\mathbf{o}_i$, $h_{\textup{CE}}(x_i)$ from \eqref{eq:pcbf} with warmstart $\mathbf{o}_{i-1}$  
            
            \IF{$\hpb(x_i)=0$}
                \STATE $\mathrm{in\_}\mathcal{S}_{\mathrm{PB},i}\leftarrow$ True
                \STATE Obtain $\mathbf{o}_i$, $h_{\textup{CE}}(x_i)$ from \eqref{eq:optprobwithnegativeslacks} with warmstart $\mathbf{o}_{i-1}$
            \ELSE
                \STATE $\mathrm{in\_}\mathcal{S}_{\mathrm{PB},i}\leftarrow$ False
            \ENDIF
            \STATE $\mathbf{s}_i \leftarrow \{(x_i, \mathbf{o}_i, \mathrm{in\_}\mathcal{S}_{\mathrm{PB},i})\}$
            \STATE $\mathbb{D}\leftarrow \mathbb{D}\cup (x_i, h_{\textup{CE}}(x_i))$
        \ENDIF
    \ENDFOR
    \STATE $\mathbb{S}\leftarrow\mathbb{S}\cup \mathbf{s}_n$
\ENDFOR
\end{algorithmic}
\end{algorithm}
\end{appendix}

\begin{IEEEbiography}[{\includegraphics[width=1in,height=1.25in,clip,keepaspectratio]{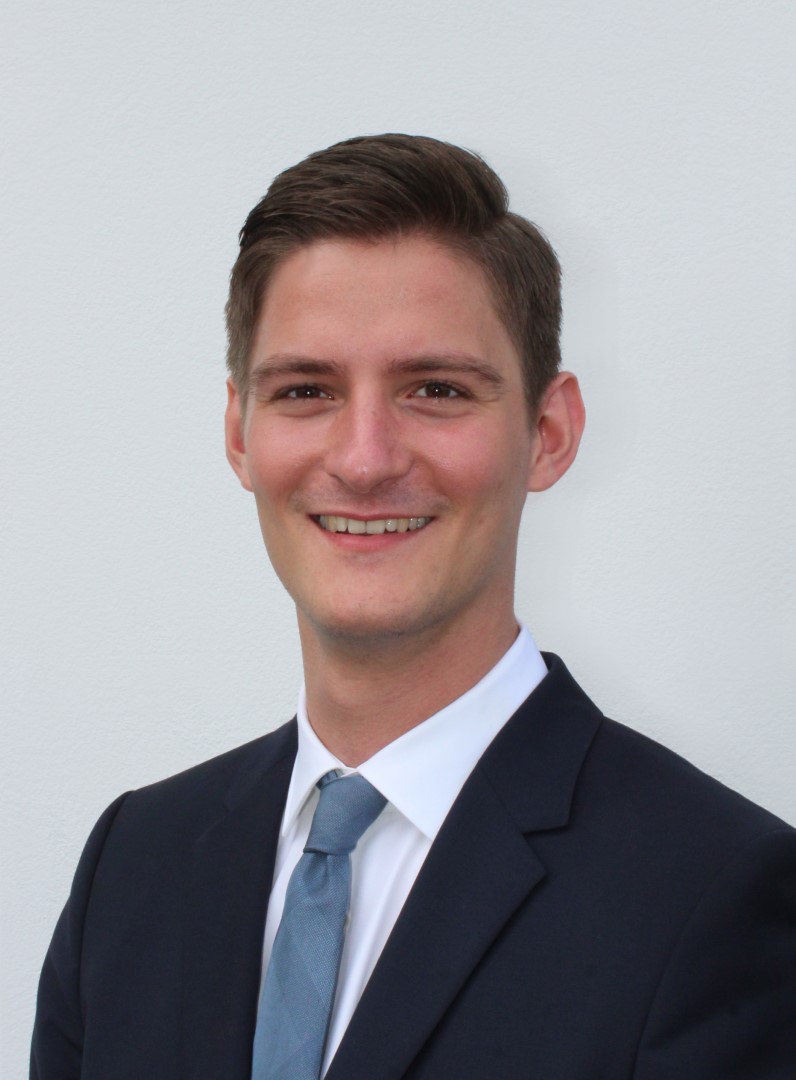}}]{Alexandre Didier} 
received his BSc. in Mechanical Engineering from ETH Zurich in 2018 and his MSc. in Robotics, Systems and Control from ETH Zurich in 2020. He is now pursuing a Ph.D. under the supervision of Prof. Dr. Melanie Zeilinger in the Intelligent Control Systems Group at the Institute of Dynamic Systems and Control (IDSC) at ETH Zurich. His research interests lie especially in Model Predictive Control and Regret Minimisation for safety-critical and uncertain systems.
\end{IEEEbiography}

\begin{IEEEbiography}[{\includegraphics[width=1in,height=1.25in,clip,keepaspectratio]{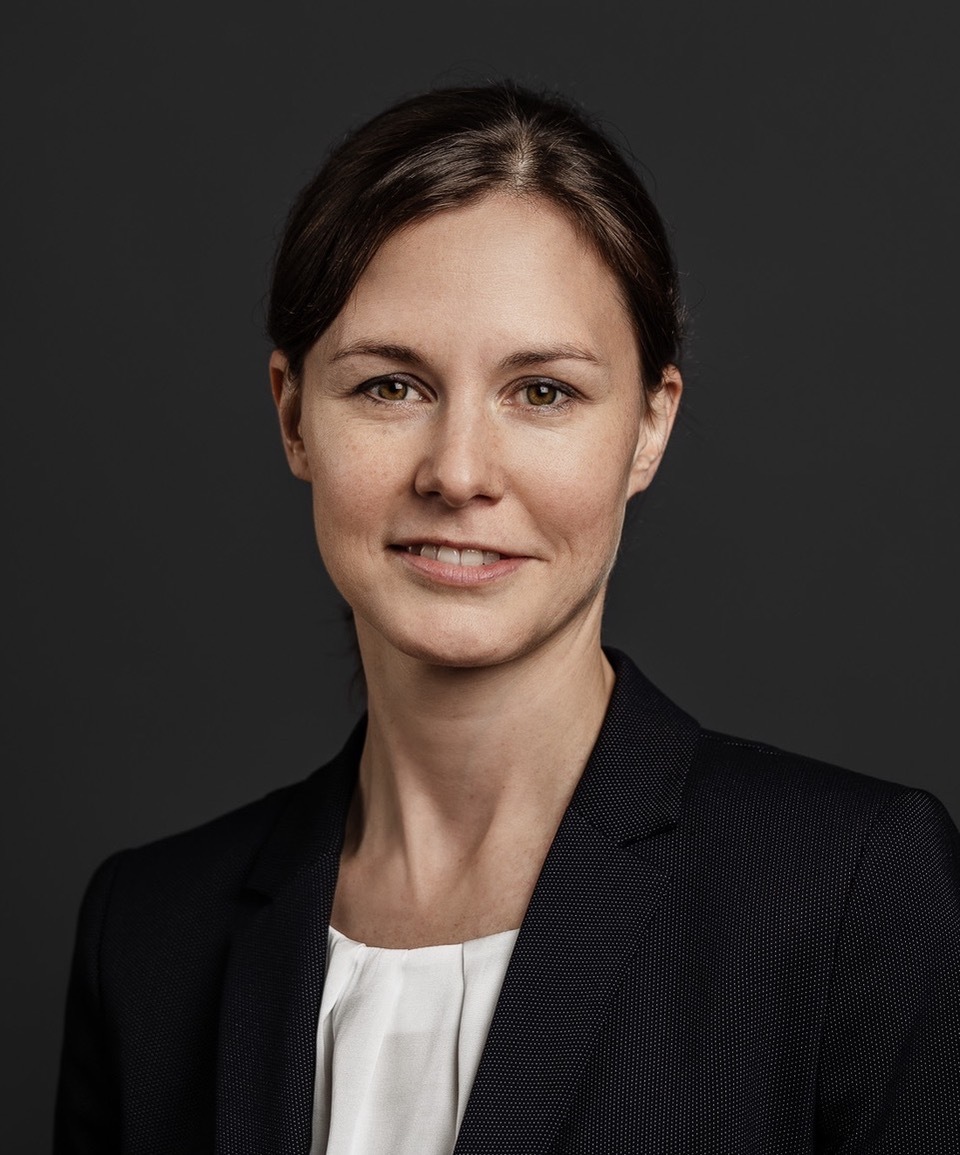}}]{Melanie N. Zeilinger} 
is an Associate Professor at ETH Zürich, Switzerland. She received the Diploma degree in engineering cybernetics from the University of Stuttgart, Germany, in 2006, and the Ph.D. degree with honors in electrical engineering from ETH Zürich, Switzerland, in 2011. From 2011 to 2012 she was a Postdoctoral Fellow with the Ecole Polytechnique Federale de Lausanne (EPFL), Switzerland. She was a Marie Curie Fellow and Postdoctoral Researcher with the Max Planck Institute for Intelligent Systems, Tübingen, Germany until 2015 and with the Department of Electrical Engineering and Computer Sciences at the University of California at Berkeley, CA, USA, from 2012 to 2014. From 2018 to 2019 she was a professor at the University of Freiburg, Germany. Her current research interests include safe learning-based control, with applications to robotics and human-in-the loop control.
\end{IEEEbiography}

\end{document}